%% file: main.tex
\documentclass[journal]{IEEEtran}
\IEEEoverridecommandlockouts

\input{myshorts}

\usepackage{algorithmic}
\usepackage[linesnumbered,ruled,vlined]{algorithm2e}
\SetKwInput{KwInput}{Input}                % Set the Input
\SetKwInput{KwOutput}{Output}              % set the Output
\usepackage{bbm}
\usepackage{amssymb, amsmath}
\newtheorem{theorem}{Theorem}
\newtheorem{lemma}{Lemma}

 \newtheorem{definition}{Definition}
 
 \newcommand{\diag}{\mathop{\mathrm{diag}}}
  
   \newtheorem{corollary}{Corollary}

 \usepackage{xcolor} 
\usepackage{lscape}
\usepackage{rangecite}
\input{psfig.sty}
 \usepackage{epsfig,subcaption,rotating, graphicx}
\begin{document}
\title{Non-Bayesian Parametric Missing-Mass Estimation}
\author{Shir Cohen, Tirza Routtenberg, \IEEEmembership{Senior Member, IEEE}, and Lang~Tong,~\IEEEmembership{Fellow,~IEEE},
		% <-this % stops a space
		\thanks{©2022 IEEE.  Personal use of this material is permitted.  Permission from IEEE must be obtained for all other uses, in any current or future media, including reprinting/republishing this material for advertising or promotional purposes, creating new collective works, for resale or redistribution to servers or lists, or reuse of any copyrighted component of this work in other works.\\
		This work is partially supported by  the ISRAEL SCIENCE FOUNDATION (ISF), grant No. 1173/16. The work of Shir Cohen  was supported under a grant from 
the  Ministry of Science and Technology of Israel.}
		\thanks{{\footnotesize{S. Cohen and T. Routtenberg are with the School of Electrical and Computer Engineering Ben-Gurion University of the Negev Beer-Sheva 84105, Israel, e-mail: shiru@post.bgu.ac.il, tirzar@bgu.ac.il.
		L. Tong is with the School of Electrical and Computer Engineering,
Cornell University, Ithaca, NY 14853 USA (e-mail:lt35@cornell.edu).}}}}% <-this % stops a space	}

\maketitle
	\nopagebreak
	\begin{abstract}
	We consider the classical problem of  missing-mass estimation, which deals with estimating the 
 total probability of unseen elements in a sample.
	The missing-mass estimation problem has various applications in machine learning, statistics,  language processing, ecology, sensor networks, and others. 
	The naive, constrained maximum likelihood (CML) estimator
	is inappropriate for this problem since it
	tends to overestimate the probability of the observed elements.
	Similarly, the conventional constrained Cram$\acute{\text{e}}$r-Rao bound (CCRB),
	which is a lower bound on the mean-squared-error (MSE) of unbiased estimators of the entire probability mass function (pmf) vector, does not provide a relevant bound on the  performance for the problem of missing-mass estimation. 
	In this paper, we introduce a
	frequentist, non-Bayesian parametric  model of the problem of missing-mass estimation. We introduce the  concept of missing-mass unbiasedness
by using the Lehmann unbiasedness definition.
We derive a non-Bayesian CCRB-type lower bound on the missing-mass MSE (mmMSE), named the  missing-mass CCRB (mmCCRB), based on the missing-mass unbiasedness.  
The proposed mmCCRB can be used for system design and for the performance evaluation of existing estimators. 
Moreover, based on the new mmCCRB, we  propose a new method to improve existing estimators by an iterative missing-mass Fisher-scoring method.
Finally, we demonstrate via numerical simulations that 
 the biased version of the  mmCCRB is a valid and informative lower bound on the mmMSE of state-of-the-art  estimators for this problem: the  CML, asymptotic profile maximum likelihood (aPML),  Good-Turing, and  Laplace estimators. We also show that the performance of the Laplace estimator is improved, in  terms of mmMSE and missing-mass bias, by using the new missing-mass Fisher-scoring method.
\end{abstract}

\section{Introduction}
\label{sec:intro} 
Given $N$ samples from a population of elements belonging to different types with unknown proportions, how should one  estimate the total probability of  unseen types? This
is a classical problem in statistics, commonly referred to as the missing-mass estimation problem
 \cite{robbins1968estimating,Good_1953}.
Missing-mass estimation has gained significant interest in various applications, such as
 ecological studies \cite{EFRON_THISTED_1976},
sensor networks \cite{Budianu_Tong_2004,Budianu_Ben_David_Tong_2006},
  machine learning, and
statistics.
In the context of 
 language processing, for example, estimation of new and existing words in text has applications such as
 language modeling, spelling correction, and word-sense disambiguation \cite{Chen_Goodman_1996,katz1987estimation}. 
Missing-mass estimation  is  especially  important for applied
problems where the sampling procedure is expensive, and the need for acquiring more data is determined by the possibility
of observing new unobserved elements. 
 
It is well known that the naive, constrained maximum likelihood (CML) estimator of the probability, i.e. 
 the empirical probability,
is ineffective if there are insufficient samples \cite{Orlitsky427,unseenVV}.
In particular, the CML estimator
assigns a zero probability to unseen events,
which does not provide relevant information on the missing mass. As a result, the constrained Cram$\acute{\text{e}}$r-Rao bound (CCRB)  \cite{Hero_constraint,Stoica_Ng_1998,Moore_Sadler_Kozick2008,Nitzan_Routtenberg_Tabrikian2019,Nitzan_Routtenberg_Tabrikian_letter}, which is  associated with the asymptotic performance of the CML estimator, is inappropriate
as a bound on the performance of missing-mass estimators
 outside the asymptotic region.
 This is because the CCRB is a lower bound on the mean-squared-error (MSE) of  the {\em{entire}} probability mass function (pmf) and not on the functional defined by the missing mass.
Various estimators of the 
missing mass  have been suggested over the years \cite{Good_1953,Orlitsky427,laplace2012pierre,Nadas_1985,gale1995good,Braess_Sauer,Krichevsky_Trofimov,Gale_Church_1994,burnham1979robust,good1956number,mao2007estimating}.
However, 
the analysis of these estimators is 
challenging and there is no comprehensive non-Bayesian estimation theory for estimating  the missing mass. In particular, there is a need for appropriate lower bounds on the MSE of the missing mass obtained by any estimator.
This theory  and these bounds are crucial for system design, 
error analysis, and performance analysis of existing estimation methods,
and for the development of new estimation methods.

\subsection{Summary of results}
In this paper, we consider the problem of estimating the missing mass, where it is assumed that we observe samples that are drawn  from an unknown distribution. First, we  introduce a non-Bayesian parametric formulation  of this estimation problem. We  use
the  missing-mass squared-error as a cost function
and derive the associated  Lehmann unbiasedness.
We develop a new non-Bayesian constrained Cram$\acute{\text{e}}$r-Rao bound (CCRB),
the missing-mass CCRB (mmCCRB), which is a lower bound on the missing-mass MSE (mmMSE)
of any  estimator with a specific Lehmann bias.
The new bound is obtained by using linear parametric constraints on the probability space and the Lehmann unbiasedness.
We investigate the properties of the mmCCRB and some special cases of this bound.
Based on the equality condition of the  mmCCRB, we  propose a new method to improve existing estimators by an iterative missing-mass Fisher-scoring method.
 The new bound is examined in simulations and compared with the performance of state-of-the-art estimators: the CML,  Good-Turing, Laplace, and asymptotic profile maximum likelihood (aPML) estimators. 
 We also show that the performance of the Laplace estimator is improved by using the new  missing-mass Fisher-scoring method.

\subsection{Related works} 
Various estimators of the 
missing mass  have been suggested in the literature. A  fundamental example is
the Good-Turing  probability estimator \cite{Good_1953},
which 
was invented  to decipher the Enigma code during  World War II. The  Good-Turing   estimator, its extensions by smoothing techniques   \cite{Nadas_1985,gale1995good}, and the Laplace estimator
 \cite{laplace2012pierre,Braess_Sauer,Krichevsky_Trofimov},
 have been shown to be  useful  for the estimation of the probability of unseen elements
 \cite{Orlitsky427,Gale_Church_1994}, and have been implemented in many  practical applications. 
  More recently, the profile maximum likelihood (PML) approach has been suggested and analyzed \cite{acharya2017unified,orlitsky2004modeling,hao2019broad} as an alternative to the CML estimator. The PML estimator is near-optimal in the MSE sense for a uniform distribution and was shown to have impressive statistical properties.

On the theoretical side, 
various works analyzed the properties of
 specific estimators. For example,
some interpretations of the
Good-Turing estimator
and its performance in terms of attenuation have been established in \cite{Orlitsky427,NIPS2015_5762}.
A derivation of the Good-Turing estimator
 from a Bayesian point of view with a uniform prior  was suggested in  \cite{Good_1953}.
%%%%%% 
Works related to the Good-Turing estimator include analysis of  its bias \cite{Good_1953,Juang_Lo_1994},
confidence intervals and convergence rate \cite{mcallester2000convergence}, and (un)consistency \cite{Wagner_Kulkarni2011}.
The performance of the Good-Turing estimator was analyzed using the theory
of large deviations in  \cite{Budianu_Tong_2004}  and the pmf of its worst-case MSE  was discussed in
\cite{Acharya_2018}.

%%%performance bounds
 Different performance bounds for the missing-mass estimation problem have been discussed in the literature.
For example,
lower and upper bounds on the expected missing mass  and
inequalities on the probability of large
deviations of  the missing mass are discussed in \cite{BEREND20121102,berend2013concentration}.
Various existing bounds are associated with the performance of a specific estimator and are distribution-free. For example, 
upper bounds on the MSE of Robbins-type estimators   have been proposed in \cite{YATRACOS1995321}.
Bounds on the worst-case MSE of the Good-Turing estimator have been developed in \cite{Acharya_2018,rajaraman2017minimax}.
On the other hand, other studies 
 provide  lower and upper bounds on the performance  of any estimator for specific distributions.
 For example, in \cite{Acharya_2018,rajaraman2017minimax} there are also
 lower and upper bounds on the minimax MSE of any estimator for specific distributions.
It should be noted that the proposed approach in this paper applies  to all algorithms/estimators and is based on evaluating the performance for each pmf value.
However, there are no  Cram$\acute{\text{e}}$r-Rao-type lower bounds on the averaged performance of any estimator of the missing mass. 
Our recent works on  estimation after selection
\cite{Routtenberg_Tong_est_after_sel,Meir_Routtenberg_journal,harel2019low,harel2021bayesian,Weiss_Routtenberg_Messer} suggest that conditional schemes,
in which the performance criterion depends on the observed data,
require different Cram$\acute{\text{e}}$r-Rao-type bounds for  analysis and system design.

%%%%%%%%%%%%%%%%%
\subsection{Organization and notation}
The remainder of the paper is organized as follows: Section
\ref{The_model} presents the  non-Bayesian parametric model  of missing-mass estimation under a multinomial model, including the conventional CCRB and constrained unbiasedness for this model and the appropriate cost function. In Section \ref{CCRBsec}, we derive a new  CCRB-type lower bound on the mmMSE. 
In Section \ref{scoring_section}, we describe the  missing-mass Fisher-scoring method.
Numerical  simulations are presented  in Section \ref{simulations_sec}. Finally, our conclusions
can be found in Section \ref{Conclusionsec}.

In the rest of this paper, vectors are denoted by boldface
lowercase letters and matrices  by boldface uppercase
letters.
The notations
${\mathbbm{1}}_{\{A\}}$ and $\Imat$ denote the indicator function of an
event $A$ and the identity matrix, respectively.
The vectors $\onevec $ and $\zerovec$   are column vectors of ones and zeros, respectively, and $\evec_m$  is the $m$th column of the identity matrix, all with appropriate dimensions. The matrix  $\diag(\avec)$ denotes the diagonal matrix with vector $\avec$ on the diagonal.
The $m$th element of the vector $\avec$,
	the $(m,q)$th element of the matrix $\Amat$,
	and the $(m_1:m_2\times q_1:q_2)$ submatrix of $\Amat$ 
	are denoted by $a_m$, $\Amat_{m,q}$, and 
	$\Amat_{m_1:m_2,q_1:q_2}$, respectively.
	The trace of a matrix $\Amat\in\mathbb{R}^{M\times M}$ is defined as
	${\text{trace}}(\Amat)=\sum_{m=1}^M \Amat_{m,m}$.
	The gradient of a vector function, $\cvec$, of $\thetavec$, $\nabla_{\thetavecsmall}\cvec$, is a matrix in $\mathbb{R}^{K\times M}$, with the $(k,m)$th element equal to $\frac{\partial c_k}{\partial \theta_m}$, where $\cvec=\left[c_1,\ldots,c_K\right]^T$ and $\thetavec=\left[\theta_1,\ldots,\theta_M\right]^T$.
For a scalar function $c$,
we denote
$\nabla_\thetavecsmall^{T}  c \define (\nabla_\thetavecsmall  c)^{T}$, and
$\nabla_\thetavecsmall^2  c \define\nabla_\thetavecsmall \nabla_\thetavecsmall^{ T}  c$.
The notations ${\rm{E}}_{\thetavecsmall} [\cdot]$ and ${\rm{E}}_{\thetavecsmall} [\cdot|A]$ represent the  expectation and conditional expectation operators,
 parametrized by a deterministic  vector, $\thetavec$,
and given the event $A$.
For a set $X$, $|X|$ represents its
cardinality. 
%%%%%%%%%%%%%%%%%
\section{Non-Bayesian  Estimation of the  
Missing Mass}
\label{The_model}
In this section, we present the problem of estimating the missing mass as a non-Bayesain parameter estimation problem. In Subsection \ref{model_sub} we describe the observation model and  the relevant probability functions. 
In Subsection \ref{CCRB_sec} we develop the $\chi$-unbiasedness and the CCRB  for estimating the unknown pmf under this model.
Finally, 
in Subsection \ref{cost_sub} we formulate the missing-mass estimation  problem, and present the missing-mass squared-error cost function, which is used  in this paper.

%%%%%%%%%%%%%%%%%%%%%%%%%%%%%%%%%%%%%%%%%%%%5
\subsection{Non-Bayesian model}
\label{model_sub}
Assume that there is a set of $M$ symbols, ${\cal{S}}=\left\{s_1,\ldots,s_M\right\}$, where
the alphabet size,
$M\geq 1$, is assumed to be finite and known. 
The elements in ${\cal{S}}$ may represent, for example,
 species in the jungle \cite{Orlitsky427},  words in a dictionary  \cite{Chen_Goodman_1996,katz1987estimation},
or  operating sensors {\cite{Budianu_Tong_2004,Budianu_Ben_David_Tong_2006}}.
The true  probability of observing symbol   $s_m$ is denoted by
$\theta_m$, $ m=1,\ldots,M$, where $\theta_m\neq 0$ for all $ m=1,\ldots,M$ and $\sum_{m=1}^{M} \theta_m = 1$.
Thus, 
 $ \thetavec\define  [\theta_1,\ldots,\theta_M]^T$ is a  pmf vector
over the discrete and finite set of symbols, ${\cal{S}}$.
As a result,  $\thetavec$  is  an element of the simplex $\Omega_\thetavecsmall$, i.e. $\thetavec\in\Omega_\thetavecsmall$,
where 
\be
\label{setworld}
\Omega_\thetavecsmall\define\left\{\thetavec\in \left[0,1\right]^M|f(\thetavec)=0\right\} ,
\ee
in which 
\be
\label{constraitfunc}
f(\thetavec) \define \sum_{m=1}^{M} \theta_m - 1.
\ee

In general constrained parameter estimation, the null-space matrix that is orthogonal to the constraints plays an important role \cite{Hero_constraint,Stoica_Ng_1998}.
For 
 the considered setting,  
the gradient of $f(\thetavec)$  w.r.t. $\thetavec$ is 
\be
\Fmat \define \nabla^T_\thetavecsmall f(\thetavec) = \onevec_M^T.
\ee
In addition, there exists a null-space matrix $\Umat \in \mathbb{R}^{M\times(M-1)}$ such that
\be
\label{Umat}
\Fmat \Umat =\onevec_M^T\Umat =  \zerovec^T ,~   \Umat^T\Umat=\Imat .
\ee
In particular, it can be verified that  $\Umat^T\yvec =\zerovec_{M-1} $ {\em{iff}} $\yvec = c\onevec_M$, where $c\in{\mathbb{R}}$ is an arbitrary constant.

Under the independent and identically distributed (i.i.d.) multinomial model  \cite{Good_1953}, it is assumed that there are
 $N$ i.i.d. samples, $\{x_n\}_{n=1}^N$, drawn according
to the  pmf described by the unknown vector, $\thetavec$.
We consider the problem of  estimating the missing mass of the unobserved symbols, which is a function of  $\thetavec$.
It can be verified that the  pmf 
of the random observation vector,
$\xvec\define[x_1,\ldots,x_N]^T\in{\cal{S}}^N$,  is a binomial distribution:
\beqna
\label{zero}
p(\xvec ;\thetavec)=\prod_{m=1}^M \theta_m^{C_{N,m}(\xvec)},~\xvec\in{\cal{S}}^N,
\eeqna
where 
\be
\label{Cnm_def}
C_{N,m}(\xvec) \define \sum_{n=1}^N {\mathbbm{1}}_{\{x_n=s_m\}},~m=1,\ldots,M,
\ee
is the number of times that the $m$th element was observed out of the $N$ samples.
Therefore, the vector $[C_{N,0}(\xvec),\ldots,C_{N,M}(\xvec)]^T$ has a multinomial distribution with parameters $N$ and $\thetavec$. 
The pmf in (\ref{zero}) can also be written as
\beqna
\label{one}
p(\xvec ;\thetavec)
=\prod_{m=1}^M \theta_m^{\sum_{r=0}^N r {\mathbbm{1}}_{\{m\in G_{N,r}(\xvec)\}}},
\eeqna
where
$G_{N,r}(\xvec)$ is  the set of elements that appear exactly $r$ times in  the $N$-length observation vector, $\xvec$. 
That is, if $m\in G_{N,r}(\xvec)$, then
$C_{N,m}(\xvec)=r$.
In particular, the set
\be
\label{GN0_def}
G_{N,0}(\xvec) \define \{m: m=1,\ldots,M, x_n\neq s_m,~\forall n=0,\ldots,N\}
\ee
is the set of elements that do not appear in the observation vector, $\xvec$,  with $C_{N,m}(\xvec)=0$.
 For example, upon observing  the  vector 
$\xvec=\left[a,c,c\right]^T$ with $N=3$ and
 ${\cal{S}}=\{a,b,c\}$,  the histogram values are $C_{3,1}(\xvec)=1$, $C_{3,2}(\xvec)=0$, and $C_{3,3}(\xvec)=2$
 according to \eqref{Cnm_def}, and
the missing mass is the pmf of $\{b\}$, since  according to \eqref{GN0_def} $G_{3,0}=\{b\}$.

Let us define the subspace of all observation vectors that do not include $s_m$ as 
\be
\label{Omega_def}
{\mathcal{A}}_m \define \{ \xvec\in{\cal{S}}^N : m \in G_{N,0}(\xvec)  \},~m=1,\ldots,M.
\ee
For a given number of measurements, $N$, the probability of the $m$th element being unobserved in these measurements is
\beqna
\label{probofkunseen}
\Pr( \xvec \in {\mathcal{A}}_m ;\thetavec) ={\rm{E}}_\thetavecsmall\left[{\mathbbm{1}}_{\left\{m\in G_{N,0}(\xvec)\right\}}\right]= (1 - \theta_m)^N,
\eeqna
$\forall m=1,\ldots,M$.
By using Bayes rule it can be seen that 
\beqna
\label{Bayes_2}
p(\xvec| \xvec \in {\mathcal{A}}_m  ;\thetavec)\hspace{-0.1cm}=\hspace{-0.1cm}
\left\{\begin{array}{lr}
\frac{p(\xvec;\thetavecsmall)}{\Pr( \xvec \in {\mathcal{A}}_m ;\thetavecsmall)}& {\text{if }} \xvec \in {\mathcal{A}}_m\\
0& {\text{otherwise}}
\end{array}\right.\hspace{-0.25cm},\hspace{-0.25cm}
\eeqna
$ m =1,\ldots, M $.
By substituting (\ref{zero}) and (\ref{probofkunseen}) in (\ref{Bayes_2}), we obtain
\beqna
\label{bayesmm}
p(\xvec | \xvec \in {\mathcal{A}}_m  ;\thetavec)
=
\left\{\begin{array}{lr}
\frac
{\prod_{l=1}^M \theta_l^{C_{N,l}(\xvec)}}{(1 - \theta_m)^N}
& {\text{if }} \xvec \in {\mathcal{A}}_m\\
0& {\text{otherwise}}
\end{array}\right.,
\eeqna
$m =1,\ldots, M $.

We denote by 
$\hat{\thetavec}:{\cal{S}}^N \rightarrow \Omega_\thetavecsmall $  an arbitrary estimator of the pmf vector, $\thetavec$, based on the observation vector, $\xvec$. 
The CML estimator of  $\thetavec$ under the parametric  constraint $f(\thetavec)=0$ from \eqref{constraitfunc} is given by
\be
\label{CML}
\hat{\theta}_m^{\text{CML}} =\frac{ C_{N,m}(\xvec)}{N}, ~m=1,\ldots,M.
\ee
In particular, the CML estimator assigns zero probability for unseen elements, i.e. for the missing mass.
Some alternative estimators are presented in Subsection \ref{special_cases_subsection}.

\subsection{CCRB and constrained unbiasedness}
\label{CCRB_sec}
In this subsection, we develop the conventional CCRB and the  unbiasedness condition for estimating $\thetavec$ under the considered model. 
The CCRB \cite{Hero_constraint,Stoica_Ng_1998}
provides a lower bound on the MSE of any locally $\chi$-unbiased estimator \cite{Nitzan_Routtenberg_Tabrikian2019,Nitzan_Routtenberg_Tabrikian_letter,sparse_con}, which is a weaker requirement than ordinary mean unbiasedness, and is defined as follows.
\begin{definition}
\label{def_chi_unbiasedness}
An estimator
$\hat{\thetavec}:{\cal{S}}^N\rightarrow\Omega_\thetavecsmall$
is said to be  a locally $\chi$-unbiased estimator 
 in the neighborhood of
$\tilde{\thetavec}\in \Omega_\thetavecsmall$  if it satisfies 
\beqna
\label{bias_MSE}
\Umat^T {\rm{E}}_{\tilde{\thetavecsmall}}
[ \hat{\thetavec} - \tilde{\thetavec}
] = \zerovec_{M-1} 
\eeqna
and
\be
\label{bias_MSE_derivative}
\left.\left\{\nabla_\thetavecsmall^T {\rm{E}}_\thetavecsmall
[ \hat{\thetavec} - \thetavec
]\right\}\right|_{\thetavecsmall=\tilde{\thetavecsmall}} \Umat = \zerovec_{M\times (M-1)} ,
\ee
where $\Umat$ is defined in (\ref{Umat}).
\end{definition}
It should be noted that 
 in this paper
the notation $\tilde{\thetavec}$ represents  a specific value (or ``local" value) of the unknown parameter vector in the simplex $\Omega_\thetavecsmall$, while $\thetavec$ is used as a general parameter in the different functions. 
For the CML estimator in \eqref{CML} we obtain that
\be
\label{zero_cml}
{\rm{E}}_\thetavecsmall
\left[\hat{\theta}^{\text{CML}}_m -\theta_m \right] =
{\rm{E}}_\thetavecsmall
\left[\frac{C_{N,m}(\xvec)}{N} -\theta_m \right] = 0,
\ee
for all $m=1,\dots,M$ and for any $\thetavec\in\Omega_\thetavecsmall$,
where the last equality follows from the mean of a variable  with a multinomial distribution.
Thus, \eqref{zero_cml} implies that the CML estimator satisfies Definition \ref{def_chi_unbiasedness} and that it is a locally $\chi$-unbiased estimator for any $\thetavec\in \Omega_\thetavecsmall$; thus,  it is a uniformly $\chi$-unbiased estimator. The  CML estimator is also 
 a  C-unbiased estimator  in the Lehmann sense  \cite{Nitzan_Routtenberg_Tabrikian2019,Lehmann}, since  the constraint $\thetavec\in\Omega_\thetavecsmall$ is  linear.

The CCRB on the MSE of any unbiased estimator in the sense of  Definition \ref{def_chi_unbiasedness} at $\tilde{\thetavec}\in\Omega_\thetavecsmall$ is given by 
\cite{Stoica_Ng_1998,Nitzan_Routtenberg_Tabrikian2019,Nitzan_Routtenberg_Tabrikian_letter,sparse_con}
\be
\label{CCRB_MSE_J}
{\rm{E}}_{\tilde{\thetavecsmall}} \left[
(\hat{\thetavec}-\tilde{\thetavec})
(\hat{\thetavec}-\tilde{\thetavec})^T
\right]
\succeq
\Umat(\Umat^T\Jmat(\tilde{\thetavec})\Umat)^{-1} \Umat^T,
\ee
where the conventional Fisher information matrix (FIM)  is 
\be
\label{Jmat_MSE}
\Jmat(\thetavec) = {\rm{E}}_\thetavecsmall
\left[ \nabla_\thetavecsmall \log p(\xvec;\thetavec) 
\nabla_\thetavecsmall^T \log p(\xvec;\thetavec)
\right],~\thetavec\in\Omega_\thetavecsmall.
\ee
In Appendix \ref{new_app_CCRB} it is shown that 
the CCRB on the trace MSE under the considered model is given by
\beqna
\label{final_CCRB1}
\sum_{m=1}^M {\rm{E}}_{\tilde{\thetavecsmall}} \left[ (\hat{\theta}_m -\tilde{\theta}_m)^2 \right]
\geq B^{\text{CCRB}} (\tilde{\thetavec}),
\eeqna
where 
\beqna
\label{final_CCRB}
B^{\text{CCRB}} (\thetavec) \define
\frac{1}{N} {\text{trace}}\left( \left(\Umat^T \left( \diag( \thetavec)
\right)^{-1} \Umat
\right)^{-1} \right).
\eeqna
However, missing-mass estimators, such as the Good-Turing and add-constant estimators, are $\chi$-biased. Thus, the performance of
these estimators should be assessed by 
 the biased CCRB \cite{ben2009constrained}, which is a function of the estimator's  bias gradient.
Moreover, the  CCRB in \eqref{final_CCRB} is a lower bound on the MSE of  estimators of the entire pmf vector, and does not provide a relevant bound on the  performance for the missing-mass estimation problem.
	This is  similar to the mismatch of the 
		naive CML estimator from \eqref{CML}, which
	tends to overestimate the probability of the observed elements.
	In the following section, we develop a new CCRB-type bound on the missing-mass estimation.

\subsection{Non-Bayesian paradigm and mmMSE risk}
\label{cost_sub}
In this subsection, we explain the rationale behind the considered approach, which is: 1) purely non-Bayesian estimation of a {\em{deterministic}} parameter; 2) an estimation of a parameter of interest in the presence of nuisance parameters that affect the accuracy of estimation; 3) based on the estimation of the entire pmf, $\thetavec$; and 4)  based on the 
mmMSE risk.

 The missing mass,
  namely the total probability mass  of the outcomes not observed in the samples  in  
  $\xvec$, is defined as
\be
\label{pr0}
p_{0}(\xvec,\thetavec)=\sum_{m=1}^M \theta_m {\mathbbm{1}}_{\{m\in G_{N,0}(\xvec)\}}.
\ee
The missing mass in \eqref{pr0} is a hybrid (mixture of random and deterministic) scalar parameter, which is a function
of both the {\em{deterministic}} pmf vector, $\thetavec$, and the {\em{random}} observation vector, $\xvec$. 
Thus, various papers in the literature (see, e.g. \cite{Acharya_2018,cohen1990admissibility}) treat the estimation problem as 
the estimation of the {\em{hybrid}} parameter, $p_{0
}(\xvec,\thetavec)$,
 which allegedly has both random and deterministic parts. However, since the random observation vector, $\xvec$, is known, the true unknown part in  $p_{0}(\xvec,\thetavec)$ is only the deterministic vector, $\thetavec$. Therefore,
 in this work we adopt the  
 non-Bayesian approach for the estimation of  {\em{deterministic}} parameters. 
 Moreover, since all the elements of  the  pmf vector, $\thetavec$, are unknown,
 we treat this estimation problem as the estimation of the  parameters of interest in \eqref{pr0} that include  the probabilities of unseen events, and refer to the other (seen) parameters in $\thetavec$ as nuisance parameters \cite{gini1996estimation,Bar_Tabrikian_PartI}.

Direct calculation of the MSE of $p_{0}(\xvec,\thetavec)$ from \eqref{pr0},
\[{\rm{E}}_\thetavecsmall\left[\left
(\sum_{m=1}^M \hat{\theta}_m {\mathbbm{1}}_{\{m\in G_{N,0}(\xvec)\}}-\sum_{m=1}^M \theta_m {\mathbbm{1}}_{\{m\in G_{N,0}(\xvec)\}}\right)^2\right],\]
requires the calculation of all
 cross-correlations of 
 estimation errors of any $\theta_m$ and $\theta_l$, $l,m\in G_{N,0}(\xvec)$.
This, in turn, requires
 computing the expectation
 of a sum of $2^M$  possible events (that represent the binary options that $m$ and/or $l$ is within/without $G_{N,0}(\xvec)$, for any $\xvec$ and any $m,l=1,\ldots,M$).
While this approach is feasible for calculating the MSE of specific missing-mass estimators (see, e.g. in \cite{Acharya_2018,rajaraman2017minimax}),  we found it 
infeasible for the calculation of the associated modified FIM and  unbiasedness, which leads to an intractable bound.

In order to
capture the relevant 
errors both meaningfully and in a way that can be easily computed,
we use here an alternative cost function, which is based on the missing-mass squared-error cost function:
\be
\label{cost_use}
C(\hat{\thetavec},\thetavec)\define \sum_{m=1}^{M} (\hat{\theta}_m - \theta_m )^2
{\mathbbm{1}}_{\{m\in G_{N,0}(\xvec) \} }  ,
\ee 
for any estimator $\hat{\thetavec}=[\hat{\theta}_1,\ldots,\hat{\theta}_M]^T$   of the pmf vector, $\thetavec=[{\theta}_1,\ldots,{\theta}_M]^T$.
The associated mmMSE
risk, which is the expected value of (\ref{cost_use}), is 
\beqna
\label{law}
{\rm{E}}_\thetavecsmall\left[C(\hat{\thetavec},\thetavec)\right] =
\sum_{m=1}^{M}{\rm{E}}_\thetavecsmall\left[ (\hat{\theta}_m - \theta_m )^2
{\mathbbm{1}}_{\{m\in G_{N,0}(\xvec) \} } \right]
\nonumber\\
= \sum_{m=1}^{M}
{\rm{E}}_\thetavecsmall\left[(\hat{\theta}_m -\theta_m)^2 | \xvec \in {\mathcal{A}}_m  \right]
\Pr(\xvec \in {\mathcal{A}}_m ;\thetavec),
\eeqna
where the last equality is obtained by using  the law of total probability and the conditional distribution from \eqref{Bayes_2}.

It can be seen that 
 in order to evaluate the 
mmMSE
 performance of the  estimator of $p_{0}(\xvec,\thetavec)$ over all possible observation vectors $\xvec$, we need to sum over the errors of all the elements of $\thetavec$ (i.e. to compute $M$ terms). This is a significant 
 reduction in computational cost compared with  the direct calculation of the MSE that requires computing the expectation over $2^M$  possible events.
In addition, since we assume that $M$ is known and finite, the sum in \eqref{cost_use} is finite.
This cost takes into account all possible estimation errors by summing over all the errors 
 in a similar manner to existing different bounds  on various cost functions  (see, e.g. in \cite{Acharya_2018,rajaraman2017minimax}) and to performance evaluation of specific estimators (see, e.g. in \cite{Juang_Lo_1994,pavlichin2019approximate}).
The use of the indicator functions in  the missing-mass squared-error cost function in \eqref{cost_use} implies that  
the error of the $m$th parameter, $\hat{\theta}_m-\theta_m$, affects the mmMSE
only for  observations $\xvec$ such that $s_m$ has not been observed. Thus, it can be seen that  the mmMSE
is the sum of the  MSEs of the parameters of interest, i.e. only the estimation
errors of   elements with the indices 
that are in $G_{N,0}(\xvec)$ from \eqref{GN0_def}.
 It should be noted that other parametric statistical analyses of the missing-mass estimation in the literature are based on the estimation of the entire pmf  vector, $\thetavec$. For example, in \cite{Juang_Lo_1994}, the full estimator of  $\thetavec$ is used to analyze the bias of missing-mass estimators. Similarly, in the development of the minimax bound in \cite{rajaraman2017minimax}, a tight bound is derived by replacing the problem of missing-mass estimation with that of distribution estimation.  Further,  some missing-mass estimators are based on estimating $\thetavec$ and then applying different smoothing approaches to obtain the estimator of the missing mass (see, e.g. \cite{Orlitsky427}).
Finally, it should be noted that the mmMSE is computed where
 the expectation is only over the randomness in the
 estimator, since the pmf is a deterministic vector in the considered model.

\section{Missing-Mass Constrained Cram$\acute{\text{e}}$r-Rao (mmCCRB) Bound}
\label{CCRBsec}
In this section, a CCRB-type lower bound is derived. 
In Subsection \ref{L_unbias_sec} we develop the uniform and local unbiasedness in the Lehmann sense under the missing-mass squared-error cost function and under the probability-space parametric constraints. In Subsection \ref{CRB_section}, we derive 
the 
 proposed bound, which is a lower
bound on the mmMSE and is a function of the Lehmann bias of the estimators.
  For the sake of generality, the  unbiasedness and the mmCCRB are first  derived for a general observation-model distribution, $p(\xvec;\thetavec)$.
  Thus, the missing-mass unbiasedness in  Subsection \ref{L_unbias_sec} and the  mmCCRB in  Subsection \ref{CRB_section} can be used for various variations of the missing-mass estimation problem, such as estimating an unknown Markov chain
from its sample \cite{berend2013concentration,hao2018learning,skorski2020missing}. 
  Then, in Subsection \ref{explicit}, we develop the closed-form mmCCRB for the classical i.i.d.  model, given in \eqref{zero},
 as well as the mmCCRB for missing-mass unbiased estimators.
Finally, in Subsection \ref{special_cases_subsection} we present some special cases of the  mmCCRB.
%%%%%%
\subsection{Lehmann unbiasedness}
\label{L_unbias_sec}
The mean-unbiasedness constraint is   commonly used in non-Bayesian parameter estimation \cite{Kay_estimation}.
Lehmann \cite{Lehmann}  proposed a generalization of the  unbiasedness concept, which is based on the
considered cost function,  as follows.
	\begin{definition}
An  estimator 
	$\hat{\thetavec}:{\cal{S}}^N\rightarrow \mathbb{R}^M$  is  an unbiased estimator of $\thetavec$ in the Lehmann sense \cite{Lehmann} w.r.t. a given cost function, $C(\hat{\thetavec},\thetavec)$,
	if 
	\be
	\label{defdef}
	{\rm{E}}_{\thetavecsmall}[C(\hat{\thetavec},\etavec) ] \geq {\rm{E}}_{\thetavecsmall}[C(\hat{\thetavec},\thetavec)],~\forall \etavec,\thetavec\in\Omega_\thetavecsmall,
	\ee
	where the simplex  $\Omega_\thetavecsmall$ is the parameter space.
		\end{definition}
The Lehmann unbiasedness definition implies that an estimator is unbiased  if, on average, 
 it is ``closer'' to  the true parameter $\thetavec$
  than to any other value in the parameter space (here, denoted by an arbitrary vector, $\etavec$).
The measure of closeness  is determined by  the considered cost function,
 $C(\hat{\thetavec},\thetavec)$. 
%%%%%%%%%%%%%%%%
Examples for Lehmann unbiasedness with different cost functions and under parametric constraints can be found in 
\cite{Nitzan_Routtenberg_Tabrikian2019,Routtenberg_Tong_est_after_sel,Meir_Routtenberg_journal,Lehmann,PCRB_J,Routtenberg_cyclic}.
The following lemma states the Lehmann unbiasedness for the estimation problem of missing-mass probability.
To this end, we define the elements of the missing-mass bias vector, $\bvec_{N,0}(\thetavec)\in\mathbb{R}^M$,  as follows:
\beqna
\label{bvec}
\left[\bvec_{N,0}(\thetavec) \right]_m \define {\rm{E}}_\thetavecsmall
\left[ (\hat{\theta}_m - \theta_m) 
{\mathbbm{1}}_{\{m\in G_{N,0}(\xvec) \} }
\right]\hspace{1.35cm}
\nonumber\\
={\rm{E}}_\thetavecsmall \left[ \hat{\theta}_m - \theta_m | \xvec \in {\mathcal{A}}_m  \right] \Pr ( \xvec \in {\mathcal{A}}_m ;\thetavec),
\eeqna
$\forall m=1,\ldots,M$, where the last equality is obtained by using 
 the law of total probability. 
%%%%%
%%%
\begin{lemma}
\label{unbiasedness_prop}
An estimator
$\hat{\thetavec}:{\cal{S}}^N\rightarrow\Omega_\thetavecsmall$
is  said to be   a {\em{uniformly}} Lehmann-unbiased estimator of $\thetavec\in \Omega_\thetavecsmall$  w.r.t. the missing-mass squared-error
cost function from (\ref{cost_use}) if
\beqna
\label{unbiasedness_cond}
\Umat^T \bvec_{N,0}(\thetavec) = \zerovec_{M-1},~\forall \thetavec\in \Omega_{\thetavecsmall},
\eeqna 
where $\Umat$   and $\bvec_{N,0}(\thetavec)$ are
defined in  \eqref{Umat} and \eqref{bvec}, respectively.
\end{lemma}
\begin{proof} The proof appears in Appendix \ref{unbiasApp}.
\end{proof}
The CCRB is a {\em{local}} bound, meaning
that it determines the achievable performance at a particular
value of $\thetavec$,  denoted here by $\tilde{\thetavec}$, based on the statistics in its neighborhood.  Similar to the local $\chi$-unbiasedness in Definition \ref{def_chi_unbiasedness}, we can define the {\em{local}}
missing-mass unbiasedness as follows.
\begin{definition}
\label{def_local_mmunbiasedness}
An estimator
$\hat{\thetavec}:{\cal{S}}^N\rightarrow\Omega_\thetavecsmall$
is said to be  a  {\em{locally}} Lehmann-unbiased estimator  \cite{Lehmann}  in the neighborhood of
$\tilde{\thetavec}\in \Omega_\thetavecsmall$ w.r.t. the missing-mass squared-error
cost function from (\ref{cost_use}) if it satisfies 
\beqna
\label{unbiasedness_cond_local}
\Umat^T \bvec_{N,0}(\tilde{\thetavec}) = \zerovec_{M-1}
\eeqna
and
\be
\label{mm_bias_MSE_derivative}
\left.\left\{\nabla_\thetavecsmall^T \bvec_{N,0}(\thetavec)\right\}\right|_{\thetavecsmall=\tilde{\thetavecsmall}}  \Umat = \zerovec_{M\times (M-1)} .
\ee
\end{definition}

It should be noted that the condition in \eqref{unbiasedness_cond} requires a {\em{uniform}} unbiasedness, for any $\thetavec\in\Omega_\thetavecsmall$, while the conditions in \eqref{unbiasedness_cond_local} and \eqref{mm_bias_MSE_derivative} are {\em{local}} conditions that are required to be satisfied only at the specific $\thetavec$, denoted here by $\tilde{\thetavec}$, for which the bound is developed. 
Both the local and the uniform missing-mass unbiasedness definitions  restrict only the values that belong to the set of unseen symbols, i.e. elements that belong to the set $G_{N,0}(\xvec)$ from \eqref{GN0_def}, in ${\cal{S}}$ to  be unbiased.
In addition, 
by comparing Definition \ref{def_chi_unbiasedness} and  Definition \ref{def_local_mmunbiasedness} it can be seen that the differences between the local
 $\chi$-unbiasedness 
and the local missing-mass unbiasedness follow from the difference of the cost functions, where both definitions use the null-space matrix, $\Umat$, which is due to the parametric constraint.
However, while there exist estimators that are $\chi$-unbiased, such as the CML estimator as shown in \eqref{zero_cml},  there are no estimators that are missing-mass unbiased in the non-asymptotic region, without splitting the data or taking  extra draws \cite{cohen1990admissibility,lo1992species,book_engen}. It is  known that even when unbiased methods do not exist in a particular
setting, such as in the case of  various nonlinear models, meaningful biased techniques with good performance can still be found \cite{ben2009constrained,Nitzan_ICASSP,yonina_penalized,kay2008rethinking}.

The uniform  missing-mass unbiasedness from (\ref{unbiasedness_cond}) can be interpreted as follows. From the definition of $\Umat$ in \eqref{Umat}, it can be verified that  $\Umat^T\yvec =\zerovec_{M-1} $ {\em{iff}} $\yvec = c\onevec_M$, where $c\in{\mathbb{R}}$ is an arbitrary constant. Thus, the condition in \eqref{unbiasedness_cond}  implies that for a uniformly Lehmann unbiased estimator, the missing-mass bias vector satisfies
\be
\label{bequal}
\bvec_{N,0}(\thetavec)=
\beta_{N,0}(\thetavec)\onevec_M,~\forall \thetavec\in\Omega_\thetavecsmall,
\ee
where  $\beta_{N,0}(\thetavec)\in{\mathbb{R}}$ is a constant. 
The condition in \eqref{bequal} is that the  $m$th element of the missing-mass bias is identical for any $m$.   This property recalls the notion of  natural estimators \cite{NIPS2015_5762}, since it assigns the same bias requirements to all symbols appearing with the same probability.

%%%%%%%%%%%%%%%%%%%%%%%%%%%%%%%%%%%
\subsection{mmCCRB}
\label{CRB_section}
Lower bounds on the mmMSE are useful for performance
analysis and system design.
In this subsection, a constrained Cram$\acute{\text{e}}$r-Rao-type
lower bound on the mmMSE from \eqref{law} is derived.
The new bound is based on the missing-mass bias in the Lehmann sense, as defined in Subsection \ref{L_unbias_sec}.
Thus, it is a bound on
the MSE of the missing mass of all estimators having a given bias function, $\bvec_{N,0}(\tilde{\thetavec})$, at
each point, $\tilde{\thetavec}\in\Omega_\thetavecsmall$.

Let us define the following 
   {\em{missing-mass Fisher information matrix (mmFIM)}} :
\beqna
\label{JJJdef}
\Jmat^{(0)}(\thetavec)\define
{\rm{E}}_\thetavecsmall\left[ \Deltamat(\xvec,\thetavec)\Deltamat^T(\xvec,\thetavec) \right],
\eeqna
	$\forall \thetavec\in\Omega_\thetavecsmall$,
where the $m$th column  of the matrix $\Deltamat(\xvec,\thetavec)\in{\mathbb{R}}^{M\times M}$  is  defined as
\beqna
\label{delta}
\Deltamat_{1:M,m}(\xvec,\thetavec) \define \nabla_\thetavecsmall \log p(\xvec|\xvec \in {\mathcal{A}}_m  ;\thetavec){\mathbbm{1}}_{\{ \xvec \in {\mathcal{A}}_m  \}},
\eeqna
 $m=1,\ldots,M$.
In addition, we define the auxiliary  matrix $\Smat(\thetavec)\in{\mathbb{R}}^{M\times M}$, in which the $m$th row  is defined as
	\beqna
	\label{Smat}
	\Smat_{m,1:M}(\thetavec) \hspace{5.75cm}
	\nonumber\\\define\bigg( \nabla^T_\thetavecsmall \bigg\{\frac{\left[\bvec_{N,0}(\thetavec)\right]_m}{\Pr(\xvec \in {\mathcal{A}}_m ;\thetavec)}
	\bigg\} + \evec_m^T \bigg) \Pr(\xvec \in {\mathcal{A}}_m ;\thetavec),
	\eeqna
	$\forall \thetavec\in\Omega_\thetavecsmall$.

We define the following  regularity condition:
\renewcommand{\theenumi}{C.\arabic{enumi}} 
\begin{enumerate}
\item
\label{cond1}
The likelihood gradient vector,
$\Deltamat_{1:M,m}(\xvec,\thetavec)$, defined in \eqref{delta},
 exists and is finite 
 $\forall\thetavec\in \Omega_\thetavecsmall$
and $\forall   m =1,\ldots, M$.
That is, 
the matrix $\Umat^T\Jmat^{(0)}(\thetavec)\Umat$
is a well-defined,  non-singular, and non-zero matrix
for any $ \thetavec\in \Omega_\thetavecsmall$.
\end{enumerate} 
\renewcommand{\theenumi}{\arabic{enumi}}
%%%%%%%%%

\begin{theorem}
	\label{theoremCRB}
Let the regularity condition \ref{cond1} be satisfied and $\hat{\thetavec}$ be an estimator of $\thetavec\in\Omega_\thetavecsmall$ with a local missing-mass bias vector in the neighborhood of $\tilde{\thetavec}\in\Omega_\thetavecsmall$ given by $\bvec_{N,0}(\tilde{\thetavec})$,  as defined
in \eqref{bvec}.
Then, the mmMSE from \eqref{law} satisfies
\be
\label{TheBound}
 {\rm{E}}_{\tilde{\thetavecsmall}}\left[C (\hat{\thetavec},\tilde{\thetavec})
 \right] \geq B^{\text{mmCCRB}} (\tilde{\thetavec}),
\ee
where  the mmCCRB evaluated at the local point, $\tilde{\thetavec}$, is 
\beqna
\label{CRB}
 B^{\text{mmCCRB}} (\tilde{\thetavec})
\define 
{\text{trace}}\left(\Smat^T(\tilde{\thetavec})\Umat(\Umat^T\Jmat^{(0)}(\tilde{\thetavec})\Umat)^{-1}\Umat^T\Smat(\tilde{\thetavec})\right)\hspace{-0.3cm}
\nonumber\\
+ \sum_{m=1}^M \frac{\left[\bvec_{N,0}(\tilde{\thetavec})\right]_m^2}{\Pr(\xvec \in {\mathcal{A}}_m ;\tilde{\thetavec})}.
\hspace{2.5cm}
\eeqna
Moreover, equality is achieved in \eqref{CRB} if
\beqna
\label{CRB_equality}
\hat{\theta}_m-\tilde{\theta}_m = \frac{[\bvec_{N,0}(\tilde{\thetavec})]_m}{\Pr (\xvec \in {\mathcal{A}}_m ;\tilde{\thetavec})}   \hspace{3.25cm}
\nonumber\\
+ \left[\Smat^T(\tilde{\thetavec})\Umat(\Umat^T\Jmat^{(0)}(\tilde{\thetavec})\Umat)^{-1} \Umat^T\Deltamat(\xvec,\tilde{\thetavec}) \right]_{m,m},
\eeqna
for any $ m=1,\ldots,M$ such that $m\in G_{N,0}(\xvec)$.
\end{theorem}
\begin{proof} The proof appears in Appendix \ref{proofCRB}.
\end{proof}
Theorem \ref{theoremCRB} provides a lower bound on the MSE of  missing-mass estimators  that have a specified bias function, $\bvec_{N,0}(\tilde{\thetavec})$. This is similar to  MSE bounds on biased estimators in the general setting \cite{yonina_penalized,kay2008rethinking}.
Biased bounds can be used to explore the fundamental tradeoff  between bias and variance,   as well as  for system design.
The specification of the biased mmCCRB
requires an {\em{a-priori}} choice of the bias gradient. The biased mmCCRB can be used for cases
where we consider an estimator with a tractable  bias gradient.
For the case of the i.i.d. model, we show in Lemma \ref{Smat_theorem} that the biased mmCCRB can be computed without the need for a bias gradient, with simple expectation terms that make the biased mmCCRB more tractable.
It should be noted that in the following we use the same notation,
$ B^{\text{mmCCRB}} (\tilde{\thetavec})$, for the mmCCRB with different bias specifications.

It can be seen that the equality condition in (\ref{CRB_equality}), which is  the requirement  for the achievability of the mmCCRB, only determines the values of the missing-mass estimation errors.
In addition, it can be seen that  the estimator defined in \eqref{CRB_equality}, $\hat{\theta}_m$, $m=1,\ldots,M$, may assign a different value for each element in the missing mass.  
That is, it is not necessarily a natural estimator  \cite{NIPS2015_5762}, in the sense that elements that appeared
the same number of  times will not necessarily get the same estimated probability.
Moreover, this estimator  is a function of the (local) unknown parameter vector, $\tilde{\thetavec}$, in the general case. Only if it is independent of $\tilde{\thetavec}$, then it is an efficient estimator and its mmMSE  
is equal to the mmCCRB.
The equality condition of the mmCCRB  in (\ref{CRB_equality}) is the basis for a new estimation method developed in Section \ref{scoring_section}.

The mmFIM in \eqref{JJJdef} is a function of the entire pmf, $\thetavec$. It can be verified that $\Jmat^{(0)}(\thetavec)$ is not a diagonal matrix (see also in \eqref{matrixFIM} below). This is because there is a coupling between the  different elements in $\thetavec$, and the estimation of one parameter  affects the accuracy in the estimation of the others.
Finally, it has been shown in recent works that   the profile likelihood, i.e. the empirical distribution up to permutation of the lexicon,
 can  considered to be a sufficient statistic for the problem of missing-mass estimation \cite{acharya2017unified,orlitsky2004modeling,hao2019broad,pavlichin2019approximate}. Thus, in general, a new bound can be developed based on the likelihood of the profile instead of the likelihood in \eqref{delta}.
 According to the extension of the data processing inequality for Fisher information \cite{zamir1998proof}, such a lower bound  may result in a tighter bound than the mmCCRB. However, it is a valid lower bound only on profile-based estimators. In addition, the derivation of such a bound is not straightforward since the entire pmf, $\thetavec$, cannot be estimated based on the profile.
Thus, we leave this topic for further investigation.

\subsection{mmCCRB for the i.i.d. model}
\label{explicit}
The mmCCRB in Theorem \ref{theoremCRB} is a lower bound on the mmMSE from \eqref{law},
which has been developed for the general observation model, $p(\xvec ;\thetavec)$,  $\thetavec\in \Omega_\thetavecsmall$. 
In this subsection, we develop the closed-form expression of the mmCCRB for the classical i.i.d. model, as described by \eqref{zero}.
In addition, we develop the mmCCRB  for
the special case of missing-mass unbiased estimators for the classical model described by \eqref{zero}.

The following corollary describes the  closed-form mmFIM for the i.i.d. case.
	\begin{corollary}
	\label{J0_lemma}
	Let the conditions of Theorem \ref{theoremCRB} be satisfied and  assume the model described in Subsection \ref{model_sub} with the observation pmf given  in \eqref{zero}.
Then, the  mmCCRB for this model is
\beqna
\label{CRB_classical}
 B^{\text{mmCCRB}} (\tilde{\thetavec})\hspace{5cm}\nonumber\\
=\frac{1}{N}
{\text{trace}}\left(\Smat^T(\tilde{\thetavec})\Umat(\Umat^T\Dmat(\tilde{\thetavec})\Umat)^{-1}\Umat^T\Smat(
\tilde{\thetavec})\right)
\nonumber\\
+ \sum_{m=1}^M \frac{\left[\bvec_{N,0}(\tilde{\thetavec})\right]_m^2}{\Pr(\xvec \in {\mathcal{A}}_m ;\tilde{\thetavec})},
\hspace{2.75cm}
\eeqna
	where $\Dmat(\thetavec)$ is a $M\times M$ diagonal matrix with the following elements on its diagonal:
	\be
	\label{D_mat}
	[\Dmat(\thetavec)]_{m,m}= -\frac{(1-\theta_m)^N}{(1-\theta_m)^2}+
	 \frac{1}{\theta_m}  \sum_{l=1,l\neq m}^M \frac{(1-\theta_l)^N}{1-\theta_l}
	    ,
	\ee
	$m=1,\ldots,M$.
 The associated equality condition is given by
\beqna
\label{CRB_equality_classical}
\hat{\theta}_m-\tilde{\theta}_m = \frac{[\bvec_{N,0}(\tilde{\thetavec})]_m}{\Pr (\xvec \in {\mathcal{A}}_m ;\tilde{\thetavec})}   \hspace{3.25cm}
\nonumber\\
+ \frac{1}{N}\left[\Smat^T(\tilde{\thetavec})\Umat(\Umat^T\Dmat(\tilde{\thetavec})\Umat)^{-1} \Umat^T\Deltamat(\xvec,\tilde{\thetavec}) \right]_{m,m},
\eeqna
for any $ m=1,\ldots,M$ such that $m\in G_{N,0}(\xvec)$,
		\end{corollary}
	\begin{proof} In Appendix \ref{FIMappendix}, it is proved that the  mmFIM from \eqref{JJJdef}
 for  the model described in Subsection \ref{model_sub} with the observation pmf in \eqref{zero} is:  
	\beqna
	\label{matrixFIM}
	\Jmat^{(0)}(\thetavec)  =
	\sum_{m=1}^M \frac{N(N-1)}{(1-\theta_m)^2} (1 - \theta_m)^N \onevec_M\onevec_M^T \hspace{1.9cm}
	\nonumber\\
	+\sum_{m=1}^M \frac{N (1-\theta_m)^N }{(1-\theta_m)^2}\left(\evec_m\onevec_M^T +
	\onevec_M\evec_m^T\right) + N\Dmat(\thetavec),
	\eeqna
 where $\Dmat(\thetavec)$ is defined in \eqref{D_mat}.
By using \eqref{matrixFIM} and the null-space property of the matrix $\Umat$ from \eqref{Umat}, $\onevec_M^T\Umat =  \zerovec^T$, we obtain
\beqna
\label{UJU_UDU}
\Umat^T\Jmat^{(0)}(\thetavec)\Umat 
= N\Umat^T\Dmat(\thetavec)\Umat. 
\eeqna
By substituting \eqref{UJU_UDU} in \eqref{CRB}, we obtain the mmCCRB for the classical model in \eqref{CRB_classical}.
Moreover,
by substituting \eqref{UJU_UDU} in \eqref{CRB_equality}, we obtain 
that
the equality condition of the mmCCRB  for the classical model
is given by \eqref{CRB_equality_classical}.
\end{proof}

The auxiliary  matrix, $\Smat(\thetavec)$, in \eqref{Smat} involves the gradient of the missing-mass bias, which makes it intractable for many  estimators with implicit bias gradient function.
	The following lemma presents a tractable form of the auxiliary  matrix for the i.i.d. case, which can be evaluated numerically. 
	This tractable form is a function of the bias of the estimator and of its correlation with the empirical histogram of the observations. 
	%%%%%
	\begin{lemma}
	\label{Smat_theorem}
	    The $m$th row of the auxiliary matrix $\Smat(\thetavec)$ from \eqref{Smat} under the model described in Subsection \ref{model_sub} with the observation pmf given  in \eqref{zero} can be calculated as
	    \beqna
	    \label{Smat_final}
	    \Smat_{m,1:M}(\thetavec) = {\rm{E}}_\thetavecsmall
	    \left[ (\hat{\theta}_m (\xvec)-\theta_m) \vvec^T(\xvec,\thetavec) {\mathbbm{1}}_{\{ \xvec \in {\mathcal{A}}_m  \}} \right]
	    \nonumber\\
		+\frac{N}{1 - \theta_m}[\bvec_{N,0}(\thetavec) ]_m(\evec_m- \onevec_M)^T,\hspace{0.5cm}
\eeqna
	    where
	   \be
	   \label{vvec}
	\vvec(\xvec,\thetavec)\define\left[\frac{C_{N,1}(\xvec)}{\theta_1},\ldots, \frac{C_{N,M}(\xvec)}{\theta_M}\right ]^T .
	\ee
	\end{lemma}
	\begin{proof} The proof appears in Appendix \ref{proofSmat}.
\end{proof}
By substituting the auxiliary  matrix, $\Smat(\thetavec)$, from Lemma \ref{Smat_theorem} in \eqref{CRB_classical}, we obtained a tractable version of the mmCCRB on the mmMSE of biased estimators.
In contrast with the traditional biased CRB, which  requires {\em{a priori}} specification of the desired bias gradient \cite{yonina_penalized,kay2008rethinking}, here, we  need to  specify the expectations in \eqref{Smat_final}, which enables an numerical calculation if needed.

In the following, we describe the mmCCRB for missing-mass unbiased estimators.
For the sake of simplicity of derivation, we assume
in the following that $\bvec_{N,0}(\thetavec)=\zerovec$. 
According to  Lemma \ref{unbiasedness_prop}, this condition  is a {\em{sufficient}} condition for the Lehmann unbiasedness in \eqref{unbiasedness_cond_local} and
\eqref{bias_MSE_derivative}. 
\begin{corollary}
	\label{lemma_unbiased}
	Let the conditions of Theorem \ref{theoremCRB} be satisfied and  assume the model described in Subsection \ref{model_sub} with   $\bvec_{N,0}(\thetavec)=\zerovec$.
Then,  the mmCCRB for this model and missing-mass unbiased estimators
  is:  
 \beqna
\label{bound_unbiased}
B^{\text{mmCCRB}} (\tilde{\thetavec})
\hspace{5.5cm}
\nonumber\\
=\frac{1}{N}\sum_{m=1}^M (1 - \tilde{\theta}_m)^{2N}
\left[\Umat(\Umat^T\Dmat(\tilde{\thetavec})\Umat)^{-1}\Umat^T\right]_{m,m},
\eeqna
 where $\Dmat(\thetavec)$ is defined in \eqref{D_mat}.
 \end{corollary}
 \begin{proof}
By substituting $\bvec_{N,0}(\thetavec)=\zerovec$
in \eqref{Smat}, we obtain that in this case $\Smat(\thetavec)$ is a diagonal matrix with the diagonal elements \beqna
\label{S_unb_gen}
[{\Smat}(\thetavec)]_{m,m} = \Pr(\xvec \in {\mathcal{A}}_m ;\thetavec)
= (1 - {\theta}_m)^{N},
\eeqna
$m=1,\ldots,M$,
where the last equality is obtained by substituting \eqref{probofkunseen}.
By substituting $\bvec_{N,0}(\thetavec)=\zerovec$ and \eqref{S_unb_gen} in   \eqref{CRB_classical}, we obtain that the mmCCRB on the mmMSE of a  
missing-mass unbiased estimator is
\beqna
\label{CRB_classical_unbiased}
 B^{\text{mmCCRB}} (\tilde{\thetavec})
=\frac{1}{N}
{\text{trace}}\left(\Umat^T\Pmat(\tilde{\thetavec})\Umat(\Umat^T\Dmat(\tilde{\thetavec})\Umat)^{-1}\right),
\eeqna
where $\Pmat(\thetavec)$ is a diagonal $M\times M$ matrix with the following elements on its diagonal:
\beqna
\label{Pmat_def}
[\Pmat(\tilde{\thetavec})]_{m,m} \define 
 (1 - \tilde{\theta}_m)^{2N}
, ~m=1,\ldots,M.
\eeqna
By substituting 
\eqref{Pmat_def} in \eqref{CRB_classical_unbiased} and using the trace operator properties,
  the mmCCRB for this case is given by \eqref{bound_unbiased}.
\end{proof}

The main advantage of the 
mmCCRB for missing-mass unbiased estimators in Corollary
	\ref{lemma_unbiased}, is that it is
	only a function of the symbol generation system via the true pmf, $\thetavec$, and the number of observations, $N$. 
	While the minimax MSE is lower-bounded by $\frac{c}{N}$ for
a constant $c$ \cite{Acharya_2018,rajaraman2017minimax}, the lower bound on the mmMSE provided by the  mmCCRB in \eqref{bound_unbiased} has a more complicated structure as a function of $N$.
It should be noted also that the minimax MSE approach is derived for a specific algorithm (e.g. Good-Turing estimator) or for the worst-case pmf, while the proposed mmCCRB applies to all algorithms and should be evaluated for each  value of $\thetavec$.
Finally, it can be seen that for each different pmf, $\thetavec$, the bound in 
\eqref{bound_unbiased} requires only the computation of the diagonal matrix   $\Dmat(\thetavec)$, which is defined in \eqref{D_mat}.

\subsection{Special cases}
\label{special_cases_subsection}
In this subsection, we develop some important special cases of  the mmCCRB   for the i.i.d. model
from Subsection \ref{explicit}.
\subsubsection{mmCCRB on the mmMSE of the CML estimator}
\label{CML_bound_sec}
	The  CML estimator from \eqref{CML} assigns a zero probability to unseen events:   \be
	\label{CML_zero}
	\hat{\theta}_m^{\text{CML}} = 0, ~\forall m\in G_{N,0}(\xvec).
	\ee
	By substituting  \eqref{CML_zero} in \eqref{bvec}, one obtains that the missing-mass bias of the CML estimator satisfies
	\beqna
	\label{bias_CML}
	[\bvec_{N,0}^{\text{CML}}(\thetavec)]_m = {\rm{E}}_\thetavecsmall [ \hat{\theta}_m^{\text{CML}} - \theta_m | \xvec \in {\mathcal{A}}_m  ] \Pr ( \xvec \in {\mathcal{A}}_m ;\thetavec)
	\nonumber\\
     = - \theta_m\Pr ( \xvec \in {\mathcal{A}}_m ;\thetavec), \hspace{2.85cm}
	\eeqna
	for any $m=1,\ldots,M$.
	By substituting \eqref{bias_CML} in
 \eqref{Smat}, we obtain that for the  CML estimator, the auxiliary  matrix satisfies $\Smat(\thetavec) = \zerovec_{M\times M}$.
 By substituting this result and \eqref{bias_CML} in  \eqref{CRB_classical}, we obtain that the mmCCRB on the mmMSE of the CML estimator (or any other estimator with the same bias function as in \eqref{bias_CML}) is
\beqna
\label{CRB_CML}
 B^{\text{mmCCRB}} (\tilde{\thetavec})
=
	\sum_{m=1}^M \tilde{\theta}_m^2 \Pr(\xvec \in {\mathcal{A}}_m ;\tilde{\thetavec}).
	\eeqna
On the other hand,
by substituting \eqref{CML_zero} in \eqref{law}, it can be seen that 
 the mmMSE of the  CML estimator evaluated at the local point, $\tilde{\thetavec}$, is
 \beqna
 \label{MSE_CML}
 {\rm{E}}_{\tilde{\thetavecsmall}} \left[ C(\hat{\thetavec}^{\text{CML}},\tilde{\thetavec}) \right] 
 = \sum_{m=1}^M \tilde{\theta}_m^2 \Pr(\xvec \in {\mathcal{A}}_m ;\tilde{\thetavec}). 
 \eeqna
Thus,  in this case, the biased mmCCRB  with the bias function of the CML estimator coincides with the mmMSE of the CML estimator. Therefore, we can conclude that there is no other estimator with the same missing-mass bias as that of the CML estimator, given in \eqref{bias_CML}, that achieves a lower mmMSE  than the CML estimator.
It should be noted that since we used the mmCCRB from Theorem \ref{theoremCRB}, which was developed for the general observation model, $p(\xvec ;\thetavec)$,  $\thetavec\in \Omega_\thetavecsmall$, this result also holds for  non-i.i.d. sampling with a general structure of $p(\xvec ;\thetavec)$ \cite{berend2013concentration,hao2018learning,skorski2020missing}.
 %%%%%
 \subsubsection{mmCCRB on the mmMSE of a general missing-mass estimator}
 \label{Good_Turing_bound}
 Various estimators of the 
missing mass have been suggested in the literature. In this paper we consider three estimators that are described in  Section \ref{simulations_sec}: The Good-Turing, Laplace, and aPML estimators. In the following, we describe how to obtain the biased mmCCRB from Corollary \ref{J0_lemma} for a general missing-mass estimator.

Consider an  estimator $\hat{p}_0(\xvec,\hat{\thetavec})$ for the missing mass from \eqref{pr0}. Under the assumption of a natural estimator \cite{NIPS2015_5762}, the associated estimator
  of a specific element in $G_{N,0}(\xvec)$ is 
\be
\label{GT2}
\hat{\theta}_m=\left\{\begin{array}{lr}\frac{ \hat{p}_0(\xvec,\hat{\thetavecsmall}) }{|G_{N,0}(\xvec)|}&{\text{if }} G_{N,0}(\xvec)\neq \emptyset\\
0& {\text{if }} G_{N,0}(\xvec)= \emptyset 
\end{array}\right.,
\ee
for any $m\in G_{N,0}(\xvec)$.
	By substituting \eqref{GT2} in \eqref{bvec}, one obtains that the missing-mass bias of an arbitrary estimator $\hat{p}_0(\xvec,\hat{\thetavec})$ satisfies
	\beqna
	\label{bias_GT}
	\left[\bvec_{N,0}(\thetavec)\right]_m\hspace{5.5cm}
	\nonumber\\ ={\rm{E}}_\thetavecsmall \left[ \frac{  \hat{p}_0(\xvec,\hat{\thetavec})}{|G_{N,0}(\xvec)|} - \theta_m | \xvec \in {\mathcal{A}}_m  \right] \Pr ( \xvec \in {\mathcal{A}}_m ;\thetavec).
	\eeqna
	While  a closed-form expression of the missing-mass bias of an arbitrary estimator in \eqref{bias_GT} is intractable, the proposed bound can be used by numerically  calculating the auxiliary matrix for this case. 
	That is, by substituting \eqref{GT2} in \eqref{Smat_final}, one obtains 
	\beqna
	\label{GTinSmat}
	\Smat_{m,1:M}(\thetavec) \hspace{5.5cm}
	\nonumber\\=
	  {\rm{E}}_\thetavecsmall
	    \left[ \left(\frac{  \hat{p}_0(\xvec,\hat{\thetavec})}{|G_{N,0}(\xvec)|}-\theta_m\right) \vvec^T(\xvec,\thetavec) | \xvec \in {\mathcal{A}}_m  \right] \Pr ( \xvec \in {\mathcal{A}}_m ;\thetavec)
	    \nonumber\\
		+\frac{N}{1 - \theta_m}[\bvec_{N,0}(\thetavec) ]_m(\evec_m- \onevec_M)^T.\hspace{0.5cm}
\eeqna
	Then, by substituting \eqref{bias_GT} and \eqref{GTinSmat} in \eqref{CRB_classical}, we obtain the associated mmCCRB, which can be evaluated numerically. In particular, a Monte Carlo approach
can be applied to approximate the expectation in  \eqref{GTinSmat}, in a similar manner to the empirical FIM approximation described in \cite{berisha2014empirical}.
%%%%%%%

%%%%%%%
\subsubsection{Uniform distribution}
\label{uniform_subsec}
For the special case where 
$\thetavec=\frac{1}{M}\onevec_M$, the diagonal elements of the matrix $\Dmat(\thetavec)$ from \eqref{D_mat}  are given by 
	\beqna
	\label{D_unif}
	[\Dmat(\thetavec)]_{m,m}= -\left(\frac{M-1}{M}\right)^{N-2}+
	M^2 \left(\frac{M-1}{M}\right)^{N}
	\nonumber\\
	=
	 M(M-2)  \left(\frac{M-1}{M}\right)^{N-2}
	    ,\hspace{1.4cm}
	    \eeqna
	$m=1,\ldots,M$.
Similarly, for this case \eqref{Pmat_def} is reduced to 
\be
\label{P_unif}
\Pmat(\thetavec) =\left(\frac{M-1}{M}\right)^{2N}\Imat_M.
\ee
By substituting  \eqref{D_unif} and \eqref{P_unif} in \eqref{CRB_classical_unbiased} and using $\Umat^T\Umat=\Imat$ from (\ref{Umat}),
we obtain that for this case the mmCCRB missing-mass unbiased estimator  is given by
\beqna
\label{bound_unbiased_unif}
 B^{\text{mmCCRB}} (\tilde{\thetavec})
=\frac{1}{N} 
\left(\frac{M-1}{M}\right)^{N+3}\frac{1}{M -2 } ,
\eeqna
for $M>2$, where we used the cyclic property of the trace. We can see that the mmCCRB for the uniform pmf from \eqref{bound_unbiased_unif} decreases as the number of samples, $N$, increases, since we have more information. 
In general, the rate of decrease is a function of $M$. For large values of $M$, the mmCCRB has approximately the order of  $\frac{1}{N}$, similar to the minimax results \cite{Acharya_2018,rajaraman2017minimax}. The lower bounds on the minimax MSE are independent of $M$, which  is assumed unknown in \cite{Acharya_2018,rajaraman2017minimax}.
As a function of $M$, the mmCCRB in \eqref{bound_unbiased_unif} increases as $M$ increases if $ 1 < M \leq \frac{N+4 + \sqrt{N^2+4}}{2}$, and  decreases as $M$ increases otherwise.
On the other hand,  for this case of uniform pmf the  CCRB in \eqref{final_CCRB} on the trace MSE is reduced to  
 \beqna
\label{final_CCRB_uniform}
B^{\text{CCRB}} (\thetavec)=\frac{1}{N}-
\frac{1}{MN}.
\eeqna
Thus, the CCRB is almost independent of the number of elements, $M$, for large value of $M N$, in contrast to the missing-mass CCRB in 
 \eqref{bound_unbiased_unif}, and, thus, is less informative for the problem of missing-mass estimation.

%%%%%%%%%%%%%%%%%%%%%%%%%%55
\section{Missing-mass Fisher-scoring-type estimation}
\label{scoring_section}
Obtaining  the minimum mmMSE estimator among all unbiased estimators  is
usually intractable. Moreover, in most nonlinear  parameter estimation problems, such an estimator does not exist. Therefore, 
in this section we describe a new iterative algorithm,
the missing-mass Fisher-scoring algorithm, that further improves the performance of existing estimators by using the proposed bound.
Similar to the Fisher-scoring method \cite{rao1973linear} and the constrained Fisher-scoring method \cite{Moore_Sadler_Kozick2008,whipps2016constrained}, the equality condition in \eqref{CRB_equality} can be used to obtain an iterative estimation procedure. In this case, the estimator at the $k$th iteration, $\hat{\thetavec}^{(k)}$, is obtained by substituting the estimator from the previous iteration,  $\tilde{\thetavec}=\hat{\thetavec}^{(k-1)}$, in \eqref{CRB_equality} to obtain
\beqna
\label{CRB_equality_estimator}
\hat{\theta}_m^{(k)}-\hat{\theta}_m^{(k-1)}
=\psi^{(k)} \left.\left\{\frac{[\bvec_{N,0}(\tilde{\thetavec})]_m}{\Pr (\xvec \in {\mathcal{A}}_m ;\tilde{\thetavec})}   \right.\right.\hspace{2.5cm}
\nonumber\\+
\left.\left.
 \left[\Smat^T(\tilde{\thetavec})\Umat(\Umat^T\Jmat^{(0)}(\tilde{\thetavec})\Umat)^{-1} \Umat^T\Deltamat(\xvec,\tilde{\thetavec}) \right]_{m,m}\right\}\right|_{\tilde{\thetavecsmall}=\hat{\thetavecsmall}^{(k-1)}},
\eeqna
for any $ m=1,\ldots,M$ such that $m\in G_{N,0}(\xvec)$,  where $\psi^{(k)}$ is the step size at the $k$th iteration.
Similarly,
for the i.i.d. model,
 the equality condition in \eqref{CRB_equality_classical} results in the following $k$th iteration:
 \beqna
\label{CRB_equality_classical_estimator}
\hat{\theta}_m^{(k)}-\hat{\theta}_m^{(k-1)}
=\psi^{(k)} \left.\left\{ \frac{[\bvec_{N,0}(\tilde{\thetavec})]_m}{\Pr (\xvec \in {\mathcal{A}}_m ;\tilde{\thetavec})}   \right.\right.\hspace{2.5cm}
\nonumber\\
\left.\left.
+ \frac{1}{N}\left[\Smat^T(\tilde{\thetavec})\Umat(\Umat^T\Dmat(\tilde{\thetavec})\Umat)^{-1} \Umat^T\Deltamat(\xvec,\tilde{\thetavec}) \right]_{m,m}\right\}\right|_{\tilde{\thetavecsmall}=\hat{\thetavecsmall}^{(k-1)}},
\eeqna
for any $ m=1,\ldots,M$ such that $m\in G_{N,0}(\xvec)$,  where $\psi^{(k)}$ is the step size at the $k$th iteration.
An appropriate step-size rule for  $\psi^{(k)}$, $k=1,2,\ldots$, should be chosen to guarantee usability (such that the resulting iterate reduces the mmMSE)  and to stabilize the
convergence.
 In comparison with the classical method of Fisher-scoring or with the constrained Fisher-scoring \cite{Moore_Sadler_Kozick2008}, the proposed iteration in  \eqref{CRB_equality_classical_estimator} is essentially a replacement of the Cram$\acute{\text{e}}$r-Rao bound (in classical Fisher-scoring) or the CCRB (in the constrained Fisher-scoring) with the new mmCCRB from \eqref{CRB_classical}  in the Fisher-scoring iteration.

The initial estimator, $\hat{\thetavec}^{(0)}$, can be
chosen to be any  existing estimator, such as the CML, Good-Turing,
Laplace, or aPML estimator, all described in Subsections \ref{simulations_sec}.
In order to obtain reasonable estimation, the initial estimator: 1) should satisfy the constraint $\thetavec\in\Omega_\thetavecsmall$ in \eqref{setworld}, i.e. $\onevec_M^T\hat{\thetavec}^{(0)}=1$; and 2) should be a natural estimator \cite{NIPS2015_5762}, i.e.   $\hat{\thetavec}^{(0)}$ assigns the same probabilities to symbols appearing the
same number of times.
Similarly,
after each iteration, we project the solution to the constraint set and to be a natural estimator, by the following steps:
 \beqna
\label{step1}
\hat{\theta}_m^{(k)}=\frac{\hat{\theta}_m^{(k)}}{\sum_{l=1}^M \hat{\theta}_l^{(k)}},
~m=1,\ldots,M
\eeqna
and
 \beqna
\label{step2}
\hat{\theta}_m^{(k)}= \frac{1}{|G_{N,C_{N,m}}(\xvec)|}
\sum_{\stackrel{l=1}{l\in G_{N,C_{N,m}}(\xvec)}}^M \hat{\theta}_l^{(k)} ,
\eeqna
$m=1,\ldots,M$.
For a desired
tolerance $\upsilon$, the algorithm exits when the condition
$||\hat{\thetavec}^{(k)}-\hat{\thetavec}^{(k-1)}|| <\upsilon$ is met.
Since the lexicon size is assumed to be known, in a case where we observe all symbols,  $G_{N,0}(\xvec)=\emptyset$, we set $\hat{p}_{0}(\xvec,\thetavec)=0$.
Finally, the algorithm  is summarized in Algorithm \ref{NR}.
\begin{algorithm}[hbt]
\DontPrintSemicolon
   \KwInput{\begin{itemize} 
   \item $M$ - Number of symbols
   \item  $\xvec$ - Observation vector
   \item $\hat{\thetavec}^{(0)}$ - Initial estimator 
   \item $\psi^{(k)}$, $k=1,\ldots$ - step sizes
   \item $\upsilon$ - Tolerance
   \item $K_{\max}$ - Maximum iteration number
   \end{itemize}
  }
  \KwOutput{$\hat{p}_{0}(\xvec,\thetavec)$ - Estimator of the missing mass}
   Initialize  $k=0$
   \\
    \eIf{$G_{N,0}(\xvec)=\emptyset$}
    {\textbf{Return:} $\hat{p}_{0}(\xvec,\thetavec)=0$}
  {Update $k\leftarrow k+1$\label{2_1}\\
  Update
$\hat{\thetavec}^{(k)}\leftarrow\hat{\thetavec}^{(k-1)}$ by
  \eqref{CRB_equality_classical_estimator} with  step size $\psi^{(k)}$ 
     \\
    Correct $\hat{\thetavec}^{(k)}$ by the projection in \eqref{step1}
    \\
    Correct $\hat{\thetavec}^{(k)}$ by the projection in \eqref{step2}
     \\
            \eIf{$||\hat{\thetavec}^{(k)}-\hat{\thetavec}^{(k-1)}|| <\upsilon$ and/or $k>K_{\max}$
            }
            {\textbf{Return:} \be
            \hat{p}_{0}(\xvec,\thetavec)=
            \frac{1}{|G_{N,0}(\xvec)|}
\sum_{m=1,m\in G_{N,0}(\xvec)}^M \hat{\theta}_m^{(k)} 
\ee.}{
            Repeat to step \ref{2_1}.
            }}
        \caption{missing-mass Fisher-scoring algorithm  for improving missing-mass estimators}\label{NR}
\end{algorithm}

%%%%%%%%%%%%%
\section{Simulations}
\label{simulations_sec}
In this section, we evaluate the proposed bound and the missing-mass Fisher-scoring method. 
In Subsection \ref{setting_subsec}, we describe the estimators and bounds that are evaluated in the simulations. In Subsections \ref{uniform2_subsec} and \ref{zipf_subsec} we evaluate the performance for uniform and Zipf distributions, respectively.

\subsection{Estimators and bounds}
\label{setting_subsec}
In the following simulations, we evaluate the performance of  four estimators of the missing mass:\\
I. CML estimator  from \eqref{CML}.\\ 
II. Good-Turing estimator -  
 The Good-Turing estimator \cite{Good_1953} of the missing mass from \eqref{pr0}
  is defined as the fraction of symbols occurring exactly once in the  observed samples divided by the length of the observation vector.
  It is well known that smoothing of the Good-Turing estimator may improve the estimation performance \cite{Good_1953,Orlitsky427}. Here we use a smooth modified version of the Good-Turing estimator,  described in   in \cite{Orlitsky427}.
This modified Good-Turing estimator
 is given by
\be
\label{GT_mm_try}
\hat{p}_{0}^{GT}(\xvec,\hat{\thetavec})=\frac{  \varphi(|G_{N,1}(\xvec)|)}{\zeta},
\ee
where
$
\varphi (t)= \max\{t,1\}$, $\forall t \in {\mathbb{R}}$,  $|G_{N,1}(\xvec)|$ is the number of elements that appear exactly once in  the $N$-length observation vector, $\xvec$, and 
\beqna
\zeta\define \varphi(|G_{N,1}(\xvec)|)\hspace{5.25cm}\nonumber\\+\sum_{r\in\{r:|G_{N,r}(\xvec)|>0\}} |G_{N,r}(\xvec)|(r+1)\frac{\varphi(|G_{N,r+1}(\xvec)|)}{|G_{N,r}(\xvec)|}
\eeqna
is a normalization factor.
\\
III.  Laplace estimator - 
 The add-constant estimator
 of the missing mass from \eqref{pr0}
  is defined as \cite{Nadas_1985}
\beqna
\label{Laplace_mm_try}
\hat{p}_{0}^{add-c}(\xvec,\hat{\thetavec})=\frac{c}{N+c(M-|G_{N,0}(\xvec)|+1)},
\eeqna
for a positive constant $c$.
The add-constant estimator has
 been applied and studied extensively and has been shown to have some optimality properties \cite{Orlitsky427}.
 For $c=1$, we obtain 
  the special case of the Laplace estimator \cite{laplace2012pierre,Nadas_1985}, used in the simulations.
  \\
IV. aPML estimator- The  PML estimator
\cite{acharya2017unified,orlitsky2004modeling,hao2019broad} has impressive statistical properties, but
is computationally challenging. Consequently,
an efficiently computable approximation for the PML distribution was proposed in \cite{pavlichin2019approximate}.
In the simulations, we use the code from \cite{aPMLcode} for computing the  aPML distributions, where the support size is used as a given input. Then, we take the smallest value in this distribution as the aPML estimator of the missing mass.

The performance of these estimators is evaluated using $500,000$ Monte-Carlo simulations
that are used to evaluate the
mmMSE,
${\rm{E}}_\thetavecsmall[C(\hat{\thetavec},\thetavec)]$, and the absolute value of the  missing-mass total bias, $|\sum_{m=1}^M \left[\bvec_{N,0}(\thetavec) \right]_m|$, as defined in \eqref{law} and \eqref{bvec}, respectively.

We  compare the performance of these estimators  with the following bounds:
\begin{itemize}
    \item The CCRB from \eqref{final_CCRB}, which is a lower bound on the MSE  of the entire pmf vector and is   presented  here  in order to compare its  behavior with that of the proposed mmMSE bounds.
    \item The mmCCRB with the bias of the CML  estimator, as given in  \eqref{CRB_CML}.
    \item  Three versions of the biased mmCCRB from  Corollary
	\ref{J0_lemma}  with the empirical bias  and  the empirical auxiliary matrix, $\Smat(\thetavec)$, as described in Subsection \ref{Good_Turing_bound}
	for the Good-Turing, Laplace, and aPML estimators. 
    \item  The  mmCCRB on missing-mass unbiased estimators from Corollary \ref{lemma_unbiased}.
\end{itemize}
It can be verified that in all the simulations the regularity condition \ref{cond1} is satisfied.

\subsection{Example 1: Uniform distribution}
\label{uniform2_subsec}
In the first experiment 
we examine  the case of a uniform pmf with equally-likely elements, i.e. where $\thetavec = \frac{1}{M}\onevec_M$, as described in Subsection \ref{uniform_subsec}. 
In Figs. \ref{Bias_uni_vs_M} and  \ref{mmMSE_uni_vs_M} we present the
missing-mass bias and the mmMSE, respectively,
of the different estimators versus the number of elements, $M$, for $N=30$. Similarly, In Figs. \ref{Bias_uni_vs_N} and  \ref{mmMSE_uni_vs_N} we present the
missing-mass bias and the mmMSE, respectively,
of the different estimators versus the number of samples, $N$, for $M=15$.
The CCRB and the mmCCRB of unbiased estimators are also presented in Figs. \ref{mmMSE_uni_vs_M} and \ref{mmMSE_uni_vs_N}.
The performance of the aPML estimator for this case is not presented  since the default of this 
estimator, where the observations are insufficient, is to choose the uniform estimator. Thus, the aPML estimator is unstable and not suitable  for comparison in the uniform distribution case.
It can be seen that for this case, the Good-Turing estimator
outperforms the two other estimators in both missing-mass bias and mmMSE terms, where the gap is larger where $M$ is larger and where $N$ is smaller.
The differences between the performance of the CML and the Laplace estimators are insignificant.
 In addition, it  can be seen in Fig. \ref{mmMSE_uni_vs_N} that  the mmMSE of the Good-Turing estimator coincides with  the proposed  mmCCRB on the mmMSE of {\em{unbiased}} estimators  from Corollary \ref{lemma_unbiased} for small values of $N$.
 While the unbiased mmCCRB is not a tight bound for this case, its curve demonstrates the influence of the system parameters, $M$ (in Fig. \ref{mmMSE_uni_vs_M}) and $N$ (in Fig. \ref{mmMSE_uni_vs_N}), on the practical mmMSE. In contrast, the  CCRB is not a useful tool for performance analysis and system design. For example, it does not reflect the influence of $M$, as also shown analytically in Subsection \ref{uniform_subsec}.
%%%%%%%%%%%%%%%%%%%%%5
\subsection{Example 2: Zipf distribution}
\label{zipf_subsec}
In the second experiment, we consider a Zipf's law distribution,
$\theta_m =\frac{m^{-s}}{\sum_{k=1}^M k^{-s}}$,  $ m =1,\ldots,M $, where $s$ is the skewness parameter.
The Zipf's law distribution is a  heavy-tailed distribution that is widely used in physical and social sciences,
linguistics, economics, and  other fields \cite{Zipf}.
In  Figs. \ref{Bias_zip1_vs_M} and \ref{mmMSE_zip_vs_M}
we present
the bias and mmMSE,
 respectively, of the different estimators
 versus the number of elements, $M$, for $N=100$ and $s=1$.
 Similarly, in  Figs. \ref{Bias_zip_vs_N} and \ref{mmMSE_zip_vs_N}
we present
the bias and mmMSE,
 respectively,
 versus the number method.sof samples, $N$, for $M=15$ and $s=1$. In addition,
in Figs. \ref{mmMSE_zip_vs_M} and \ref{mmMSE_zip_vs_N} we also present the CCRB, unbiased mmCCRB, and the  different biased mmCCRBs associated with the considered estimators.

It
can be seen in Figs. \ref{Fig_zif_vs_M}-\ref{Fig_zif_vs_N} that the CML estimator has the largest missing-mass bias and the largest mmMSE.
Additionally, the   aPML estimator has  the smallest missing-mass bias and the smallest mmMSE for all $N$ and $M$.
In Figs. \ref{mmMSE_zip_vs_M} and \ref{mmMSE_zip_vs_N} we can see that the mmCCRB on the  mmMSE of missing-mass unbiased estimators is  lower than the actual mmMSE of theestimators, since 
 these estimators are {\em{biased}} in the Lehmann sense. However,
the  curve of the mmCCRB demonstrates the influence of the system parameters, $M$ (in Fig. \ref{mmMSE_zip_vs_M}) and $N$ (in Fig. \ref{mmMSE_zip_vs_N}), on the practical mmMSE. In contrast, the  CCRB does not reflect the influence of $M$ also in this case.
 
\begin{figure}[hbt]
     \centering
\subcaptionbox{\label{Bias_uni_vs_M}}[\linewidth]
{ \includegraphics[width=8.5cm]{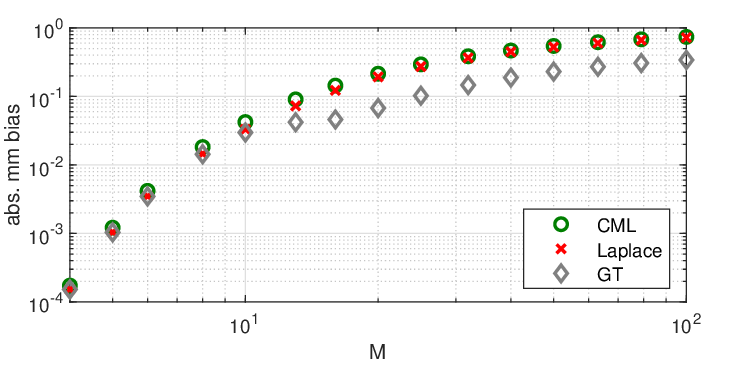}}
\subcaptionbox{\label{mmMSE_uni_vs_M}}[\linewidth]
{\includegraphics[width=9cm]{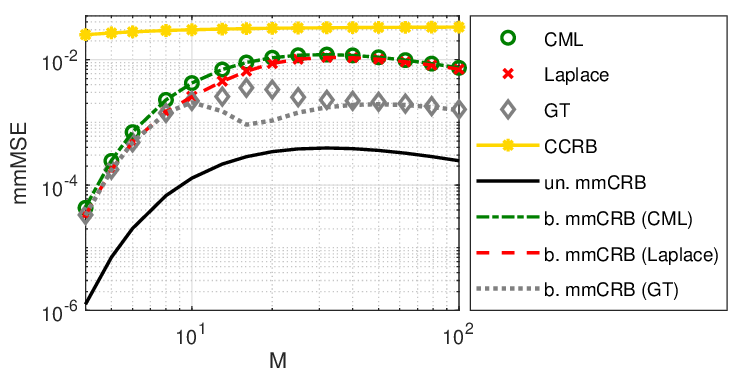}} 
     \caption{Example 1 (uniform pmf): The performance of the  CML,  Good-Turing, and Laplace estimators   versus the number of elements, $M$, in terms of 
      missing-mass bias (a) and the 	mmMSE (b). In (b) we also present the CCRB and the biased and unbiased versions of the  mmCCRB.}
 \label{Fig_uni_vs_M}
 \end{figure}
  \begin{figure}[hbt]
     \centering
\subcaptionbox{\label{Bias_uni_vs_N}}[\linewidth]
{\includegraphics[width=8.5cm]{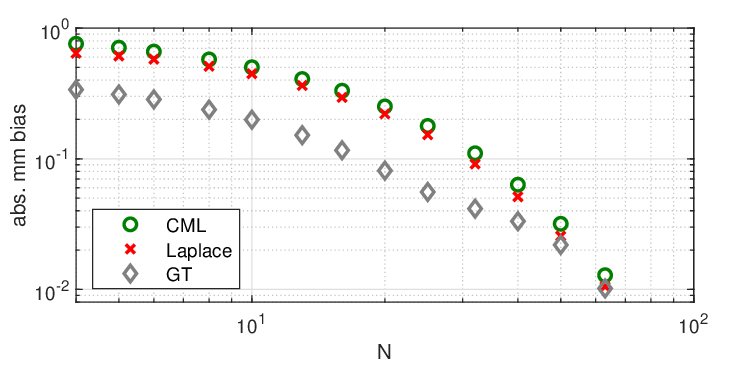}} 
\subcaptionbox{\label{mmMSE_uni_vs_N}}[\linewidth]
{\includegraphics[width=9cm]{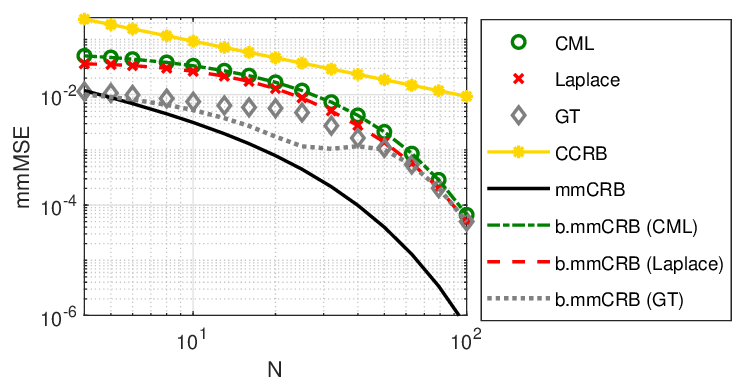}} 
     \caption{Example 1 (uniform pmf): The performance of the  CML,  Good-Turing, and Laplace estimators   versus the number of samples, $N$, in terms of 
      missing-mass bias (a) and the 	mmMSE (b). In (b) we also present the CCRB and the biased and unbiased versions of the  mmCCRB.}
 \label{Fig_uni_vs_N}
 \end{figure}
It can be seen 
that the biased mmCCRB with the CML bias coincides with the mmMSE of the CML estimator, as shown analytically in \eqref{CRB_CML}-\eqref{MSE_CML}.
Similarly, the biased mmCCRB associated with the Laplace estimator coincides with the mmMSE of the Laplace estimator. Thus, the CML and Laplace estimators achieve the lowest mmMSE for their associated bias function. However, for the Good-Turing and aPML estimators there is a gap between the associated mmCCRB and the mmMSE.  As the sample size, $N$, increases, the
mmMSE of these estimators achieves the associated mmCCRBs. 
\begin{figure}[hbt]
     \centering
\subcaptionbox{\label{Bias_zip1_vs_M}}[\linewidth]
{\includegraphics[width=8.5cm]{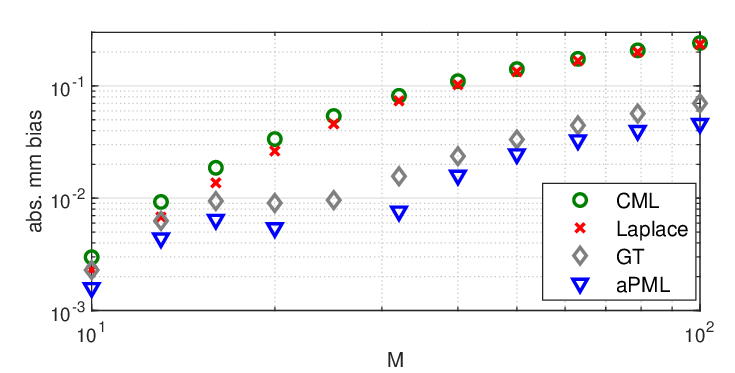}}
\subcaptionbox{\label{mmMSE_zip_vs_M}}[\linewidth]
{\includegraphics[width=9cm]{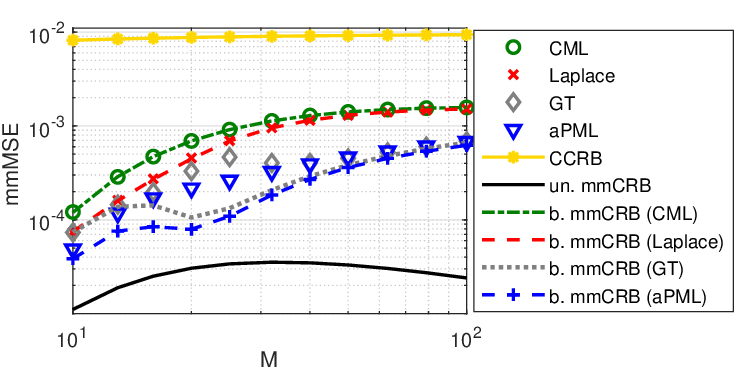}} 
  \caption{Example 2  (Zipf distribution): The performance of the  CML,  Good-Turing,  Laplace, and aPML estimators   versus the number of elements, $M$, for $N=100$ and $s=1$ in terms of 
      missing-mass bias (a) and the 	mmMSE (b). In (b) we also present the CCRB and the biased and unbiased versions of the  mmCCRB.}
 \label{Fig_zif_vs_M}
 \end{figure}
Thus,  in this case, there can be estimators with the same bias as the Good-Turing or aPML  estimator but with a lower mmMSE. For large values of $M$ and small values of $N$, the mmMSE of the aPML and the Good-Turing  estimators coincides. These regions can also be deduced from a comparison between the biased mmCCRBs associated with these estimators.
  \begin{figure}[hbt]
     \centering
\subcaptionbox{\label{Bias_zip_vs_N}}[\linewidth]
{ \includegraphics[width=8.5cm]{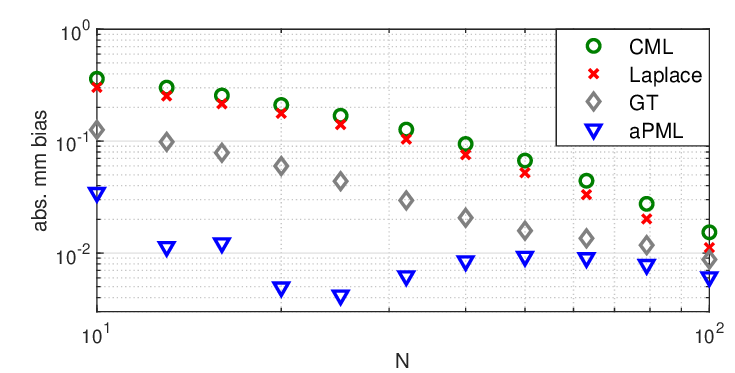}}
\subcaptionbox{\label{mmMSE_zip_vs_N}}[\linewidth]
{\includegraphics[width=9cm]{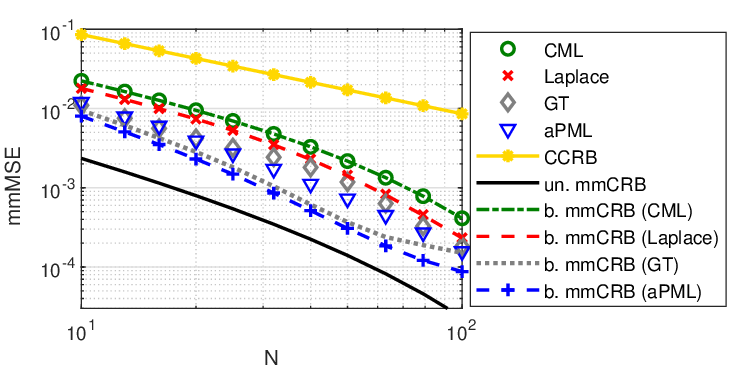}} 
     \caption{Example 2  (Zipf distribution): The performance of the  CML,  Good-Turing,  Laplace, and aPML estimators   versus the number of samples, $N$, for $M=15$ and $s=1$ in terms of 
      missing-mass bias (a) and the 	mmMSE (b). In (b) we also present the CCRB and the biased and unbiased versions of the  mmCCRB.}
 \label{Fig_zif_vs_N}
 \end{figure}
 
 In Figs. \ref{iteration1} and \ref{iteration2}, we compare  the missing-mass bias and the mmMSE of the Laplace estimator and the  estimators that are obtained after $1-5$ iterations of the proposed missing-mass Fisher-scoring method in Algorithm \ref{NR} with  initialization by the Laplace estimator. We set $\psi^{(k)}=\frac{1}{N}$, $\forall k\geq 1$ in Algorithm \ref{NR}.
 It can be seen that the proposed missing-mass Fisher-scoring method reduces the missing-mass bias and the mmMSE of the Laplace estimator. In addition, the proposed method is consistent in the sense that by using more iterations, we obtain better estimators.
These results demonstrate that the proposed mmCCRB 
 can be used to  further improve the performance of existing estimators of the missing mass.
   \begin{figure}[htb]
     \centering
\subcaptionbox{\label{iteration1}}[\linewidth]
{\includegraphics[width=5.5cm]{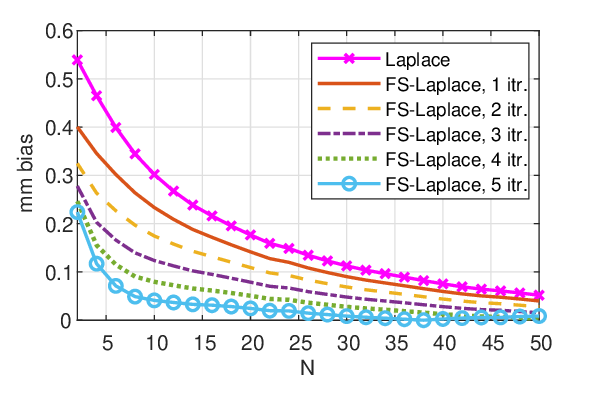}} 
\subcaptionbox{\label{iteration2}}[\linewidth]
{\includegraphics[width=5.5cm]{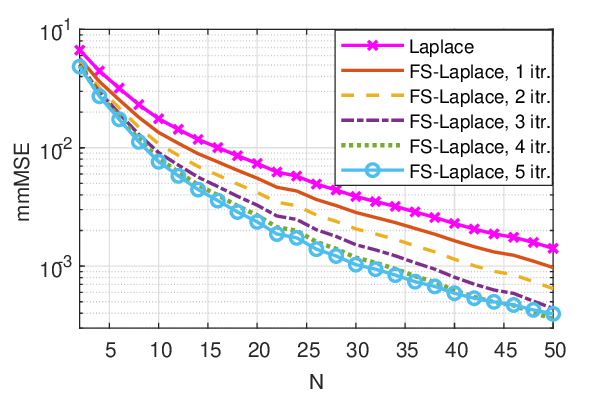}}
     \caption{Example 2 (Zipf distribution): The performance of the Laplace estimator and its improvement by the missing-mass  Fisher-scoring method (from Algorithm \ref{NR}) after 1-5 iterations  versus the number of samples, $N$, in terms of 
      missing-mass bias (a) and the mmMSE (b). 
      }
 \label{iteration}
 \end{figure}

%%%%%%%%%%%%%%%%

\section{Conclusion}
\label{Conclusionsec}
In this paper, we consider the problem of estimation
of the missing mass.
Similar to the 
		naive CML estimator, which
	overestimates the probability of the observed elements, the CCRB on the MSE of  $\chi$-unbiased 
estimators does not provide a relevant bound  for the missing-mass estimation problem.
Hence, we adopt a new non-Bayesian approach, which is based on using
 the mmMSE risk function that only penalizes the estimation errors of elements that belong to the  missing mass. 
The missing-mass unbiasedness, which   is based on Lehmann's concept of unbiasedness and the  mmMSE risk function, is proposed.
We develop a new Cram$\acute{\text{e}}$r-Rao-type bound for this problem, the mmCCRB, which is a lower bound
on the mmMSE of any locally  missing-mass unbiased estimators. In addition,
the biased mmCCRB on the mmMSE of for missing-mass biased estimators is developed. 
By using the mmCCRB on the mmMSE of the CML estimator, we show analytically that  the CML estimator has the smallest mmMSE among all estimators that have the same missing-mass bias as the CML estimator. 
Based on the equality condition of the new mmCCRB, we  derive a new method to improve existing estimators by an iterative missing-mass Fisher-scoring method.

In the simulations, we show  that 
 the unbiased mmCCRB is not a tight bound, but it can predict the behavior of the estimators w.r.t. the different system parameters.  The biased versions of the mmCCRB are tight for the CML and Laplace estimators. Thus,  the CML and the Laplace estimators achieve the lowest mmMSE possible for their bias. In contrast, for the Good-Turing and aPML estimators, there are regions in which one can theoretically find a better estimator (in the mmMSE sense)  with the same bias.
In addition, the different biased mmCCRBs 
can also be useful for setting the
 order relation between the different estimators and for exploring the bias-variance tradeoff. It is also shown that  the proposed missing-mass Fisher-scoring method reduces  the missing-mass bias and the mmMSE of the Laplace estimator. 
Future work should include  missing-mass estimation with an unknown and infinite alphabet size.
In addition, further investigation is needed 
regarding development of  lower bounds  based on the profile likelihood   that  may be tighter  than the  mmCCRB.     
%%%%%%%%%%%%%%%%%%%%%%%%%%%
\appendices
\section{Derivation of \eqref{final_CCRB}}
\label{new_app_CCRB}
In this appendix, we develop the CCRB on the MSE for pmf estimation.
First, we note that taking the logarithm of \eqref{zero} yields 
the following log-likelihood function:
\beqna
\label{log_zero}
\log p(\xvec ;\thetavec)=\sum_{m=1}^M C_{N,m}(\xvec) \log \theta_m,~\xvec\in{\cal{S}}^N.
\eeqna
By substituting the derivative of  \eqref{log_zero}  w.r.t. $\thetavec$ in \eqref{Jmat_MSE}, we obtain that 
the $(l,m)$ element of 
the FIM  is given by
\beqna
\label{FIM_MSE}
\left[ \Jmat(\thetavec) \right]_{l,m}
&=&
\frac{1}{\theta_l\theta_m}
{\rm{E}}_\thetavecsmall \left[
C_{N,l}(\xvec)C_{N,m}(\xvec)
\right]
\nonumber\\
&=&
\begin{cases}
N(N - 1)  & m\neq l \\
N^2 + N \frac{1-\theta_m}{\theta_m}  & m = l
\end{cases} ,
\eeqna
$\forall m,l=1,\ldots,M$, where  
$C_{N,m}(\xvec)$, $m=1,\ldots,M$, are defined in \eqref{Cnm_def} and the
last equality holds by using known results on the moments of the multinomial distributed variables \cite{10.2307/2237124}.
By using the elements in \eqref{FIM_MSE}, the FIM for the probability estimation model  can be written in a matrix form as
\be
\label{Jmat_MSEE}
\Jmat (\thetavec) = 
N(N-1)\onevec_M\onevec_M^T + N \left( \diag( \thetavec)
\right)^{-1}.
\ee
 It should be noted that 
$\Jmat (\thetavec)$ from \eqref{Jmat_MSEE} is a well-defined, non-singular matrix, since we assume that $\theta_m \neq 0$, $\forall m=1,\dots,M$.
By substituting (\ref{Jmat_MSEE}) in (\ref{CCRB_MSE_J}), one obtains the following closed-form CCRB on the MSE  under the constraint $\thetavec\in\Omega_\thetavecsmall$:
\beqna
\label{CCRB_reg}
{\rm{E}}_{\tilde{\thetavecsmall}} \left[
(\hat{\thetavec}-\tilde{\thetavec})
(\hat{\thetavec}-\tilde{\thetavec})^T
\right]\hspace{3.5cm}\nonumber\\
\succeq
\frac{1}{N(N-1)}\Umat(\Umat^T \onevec_M \onevec_M^T
\Umat)^{-1} \Umat^T
\nonumber\\
+ \frac{1}{N} \Umat\left(\Umat^T \left( \diag( \tilde{\thetavec})
\right)^{-1}
\Umat\right)^{-1} \Umat^T
\nonumber\\
= \frac{1}{N} \Umat\left(\Umat^T \left( \diag( \tilde{\thetavec})
\right)^{-1}
\Umat\right)^{-1} \Umat^T,
\eeqna
where the last equality is obtained by substituting   (\ref{Umat}), which implies $\onevec_M^T\Umat=\zerovec_{M-1}$. 
By applying the trace operator on the CCRB from \eqref{CCRB_reg} and using $\Umat^T\Umat=\Imat$ from (\ref{Umat}) we obtain the bound on the trace MSE in \eqref{final_CCRB1}-\eqref{final_CCRB}.
	\section{Proof of Lemma \ref{unbiasedness_prop}}
	\label{unbiasApp}
	In this appendix, we develop the missing-mass Lehmann unbiasedness.
By substituting the missing-mass 
squared-error cost function from  \eqref{cost_use}
and $\Omega_\thetavecsmall$ from \eqref{setworld} in (\ref{defdef}), one obtains that the Lehmann-unbiasedness condition for the missing-mass estimation problem is given by
\beqna
\label{opt_prob}
\sum\nolimits_{m=1}^M
{\rm{E}}_{\thetavecsmall}\left[(\hat{\theta}_m- \eta_m)^2  
{\mathbbm{1}}_{\{ m \in G_{N,0} (\xvec) \}} \right]
\hspace{1.5cm}\nonumber\\
\geq
\sum\nolimits_{m=1}^M
{\rm{E}}_{\thetavecsmall}\left[(\hat{\theta}_m- \theta_m)^2  
{\mathbbm{1}}_{\{ m \in G_{N,0} (\xvec) \}} \right] ,
\eeqna
$\forall \thetavec, \etavec \in \Omega_\thetavecsmall$. By using the definition of the constrained set in \eqref{setworld} and
 since $\Umat$ from \eqref{Umat} is the null-space matrix of this constrained set, 
	then, for a given  $\thetavec\in\Omega_\thetavecsmall$,
any $\etavec\in\Omega_\thetavecsmall$ can be written as (see, e.g. Section 4.2.4 in \cite{Boyd_2004})
\be
\label{CConection}
\etavec = \thetavec + \Umat\wvec,
\ee
where $\wvec\in \mathbb{R}^{M-1}$ is an arbitrary vector.
By substituting (\ref{CConection}) in (\ref{opt_prob}), we obtain
\beqna
\label{unbias_app3}
\sum\nolimits_{m=1}^M
{\rm{E}}_{\thetavecsmall}\left[(\hat{\theta}_m- \theta_m-\evec_m^T\Umat\wvec))^2  
{\mathbbm{1}}_{\{ m \in G_{N,0} (\xvec) \}} \right]
\nonumber\\
\geq
\sum\nolimits_{m=1}^M
{\rm{E}}_{\thetavecsmall}\left[(\hat{\theta}_m-\theta_m)^2  
{\mathbbm{1}}_{\{ m \in G_{N,0} (\xvec) \}} \right],
\eeqna
$\forall\thetavec\in\Omega_\thetavecsmall, \wvec\in\mathbb{R}^{M-1}$.
By using \eqref{probofkunseen}, the unbiasedness condition from \eqref{unbias_app3} can be rewritten as:
\beqna
\label{unbiasedness_app4}
\sum\nolimits_{m=1}^M (\evec_m^T\Umat\wvec)^2
\Pr (\xvec \in {\mathcal{A}}_m ;\thetavec) 
\hspace{3cm}
\nonumber\\
\geq
2 \sum\nolimits_{m=1}^M {\rm{E}}_\thetavecsmall 
\left[ (\hat{\theta}_m- \evec_m^T\thetavec)
{\mathbbm{1}}_{\{ m \in G_{N,0} (\xvec) \}}
\right]
\evec_m^T\Umat\wvec,
\eeqna
$\forall\thetavec\in\Omega_\thetavecsmall, \wvec\in\mathbb{R}^{M-1}$.
Since the condition in \eqref{unbiasedness_app4} should be satisfied for any $\wvec\in{\mathbb{R}}^{M-1}$,
it should be satisfied in particular for both $\wvec=\epsilon\evec_k$
and $\wvec=-\epsilon\evec_k$, where $\epsilon>0$.
By summing the separate substitution of $\wvec=\pm\epsilon\evec_k$ (that is, the result of substituting $\wvec=\epsilon\evec_k$ into \eqref{unbiasedness_app4} and the result of substituting $\wvec=\epsilon\evec_k$ into the same equation), 
we obtain the following  {\em{necessary}} condition for \eqref{unbiasedness_app4} to hold:
\be
\label{unbiasedness_app5}
\sum_{m=1}^M {\rm{E}}_\thetavecsmall\left[(\hat{\theta}_m - \theta_m)
{\mathbbm{1}}_{\{ m \in G_{N,0} (\xvec) \}}
\right]
\evec_m^T \Umat = \zerovec_{M-1}^T,
\ee
$\forall\thetavec\in\Omega_\thetavecsmall$.
Since the l.h.s. of \eqref{unbiasedness_app4}  is a quadratic term,  it can be verified that \eqref{unbiasedness_app5} is also a {\em{sufficient}} condition for unbiasedness in this case.  
Therefore, by applying the transpose operator on \eqref{unbiasedness_app5}, one obtains that the missing-mass unbiasedness in \eqref{unbiasedness_cond} is the Lehmann unbiasedness under the missing-mass squared error cost function.

\section{Proof of Theorem \ref{theoremCRB}}
\label{proofCRB}
In this appendix, we develop  the new mmCCRB from Theorem \ref{theoremCRB}. The proof is divided into: 1) the development of Lemma \ref{lemma1} in Subsection \ref{lemma1_sec}; 2) the main development of the bound based on the covariance inequality in Subsection \ref{cov_subsec}; and 3) derivation of the equality condition in Subsection \ref{equality_cond_proof}.
To this end, we define $\Gammamat(\xvec,\thetavec)$ as a diagonal $M\times M$ matrix with the following elements on its diagonal:
\be
\label{Gamma_def}
[\Gammamat(\xvec,\thetavec)]_{m,m}\define\epsilon_m (\thetavec){\mathbbm{1}}_{\{ m \in G_{N,0} (\xvec) \}},~ m=1\ldots,M,
\ee
where
\be
\label{epsilom}
\epsilon_m(\thetavec) \define \hat{\theta}_m - {\rm{E}}_\thetavecsmall \left[\hat{\theta}_m| \xvec \in {\mathcal{A}}_m \right],~m=1,\ldots,M.
\ee
\subsection{Lemma \ref{lemma1}}
\label{lemma1_sec}
In this subsection, we  prove the following Lemma:
	\begin{lemma}
	\label{lemma1}
	\beqna
		\label{52}
	{\rm{E}}_\thetavecsmall \left[ \Gammamat(\xvec,\thetavec) \Deltamat^T(\xvec,\thetavec)
	\right]
	= 
	 \Smat (\thetavec),
	\eeqna
	%%%%%%%%%%
	where $\Deltamat(\xvec,\thetavec)$,  $\Smat(\thetavec)$, and $\Gammamat(\xvec,\thetavec)$ are defined in \eqref{delta}, \eqref{Smat}, and \eqref{Gamma_def}, respectively.
	\end{lemma}
	%%%%%
\begin{proof}
	By substituting \eqref{delta},
	\eqref{Gamma_def}, and \eqref{epsilom} in \eqref{52}
one obtains that the $m$th row of $	{\rm{E}}_\thetavecsmall \left[ \Gammamat (\xvec,\thetavec)\Deltamat^T(\xvec,\thetavec)
	\right]$ on the r.h.s. of \eqref{52} satisfies
\beqna
	\label{prod}
	{\rm{E}}_{\thetavecsmall} \left[ \epsilon_m(\thetavec)  \Deltamat_{1:M,m}^T(\xvec,\thetavec)
	\right]\hspace{4.25cm}
	\nonumber\\
	=	{\rm{E}}_{\thetavecsmall} \left[ \epsilon_m(\thetavec) \nabla_\thetavecsmall \log p(\xvec|\xvec \in {\mathcal{A}}_m  ;\thetavec){\mathbbm{1}}_{\{ \xvec \in {\mathcal{A}}_m  \}}\right]\hspace{1.25cm}
	\nonumber\\= 
	\Pr (\xvec \in {\mathcal{A}}_m ;\thetavec)\hspace{-0.2cm}
	\sum_{\alphavecsmall\in{\mathcal{A}}_m } \epsilon_m(\thetavec) \nabla_\thetavecsmall^T  \Pr(\xvec=\alphavec| \xvec \in {\mathcal{A}}_m ;\thetavec), 
\eeqna
 $ m =1,\ldots, M$,
where the last equality stems from 
\[\nabla_\thetavecsmall \log p(\xvec|\xvec \in {\mathcal{A}}_m  ;\thetavec)=\frac{\nabla_\thetavecsmall  p(\xvec|\xvec \in {\mathcal{A}}_m  ;\thetavec)}{p(\xvec|\xvec \in {\mathcal{A}}_m ;\thetavec)},~\forall \xvec \in {\mathcal{A}}_m ,
\]
in which ${\mathcal{A}}_m$ is defined in \eqref{Omega_def}.
Then, by applying the product rule on  the r.h.s. of \eqref{prod}, we obtain
\beqna
\label{51}
	\sum_{\alphavecsmall\in{\mathcal{A}}_m } \epsilon_m(\thetavec) \nabla_\thetavecsmall^T  \Pr(\xvec=\alphavec| \xvec \in {\mathcal{A}}_m ;\thetavec)	\hspace{2cm}
	\nonumber\\= \nabla_\thetavecsmall^T \left\{ \sum_{\alphavecsmall\in{\mathcal{A}}_m } \epsilon_m(\thetavec)
	\Pr\left(\xvec=\alphavec|\xvec \in {\mathcal{A}}_m ;\thetavec\right)\right\}\hspace{0.5cm}
	\nonumber\\
	- \sum_{\alphavecsmall\in{\mathcal{A}}_m} \nabla_\thetavecsmall^T\{ 
	\epsilon_m(\thetavec)\} \Pr(\xvec=\alphavec| \xvec \in {\mathcal{A}}_m  (\xvec);\thetavec)\nonumber\\
	=\nabla_\thetavecsmall^T \left\{ 
	{\rm{E}}_\thetavecsmall \left[
	\epsilon_m(\thetavec) | \xvec \in {\mathcal{A}}_m 
	\right]\right\}\hspace{3cm}
	\nonumber\\
	- \sum_{\alphavecsmall\in{\mathcal{A}}_m} \nabla_\thetavecsmall^T\{ 
	\epsilon_m(\thetavec)\} \Pr(\xvec=\alphavec| \xvec \in {\mathcal{A}}_m  (\xvec);\thetavec).
	\eeqna
	Computing the conditional expectation of \eqref{epsilom}, given the event that $m \in G_{N,0}(\xvec)$, results in
\be
\label{epsilonzero}
{\rm{E}}_\thetavecsmall\left[\epsilon_m(\thetavec)| \xvec \in {\mathcal{A}}_m  \right] =0,~m=1,\ldots,M.
\ee
In addition, computing the gradient of  \eqref{epsilom} results in
	\be
\label{epsilom2}
\nabla_\thetavecsmall^T 
	\epsilon_m(\thetavec) = - \nabla_\thetavecsmall^T {\rm{E}}_\thetavecsmall \left[\hat{\theta}_m|\xvec \in {\mathcal{A}}_m \right],
\ee
$m=1,\ldots,M$.
 By using the conditional expectation definition, and then substituting \eqref{epsilonzero} and  \eqref{epsilom2} in \eqref{51}, one obtains
	\beqna
	\label{51_B}
		\sum_{\alphavecsmall\in{\mathcal{A}}_m } \epsilon_m(\thetavec) \nabla_\thetavecsmall^T  \Pr(\xvec=\alphavec| \xvec \in {\mathcal{A}}_m  (\xvec);\thetavec)	\hspace{1.8cm}
	\nonumber\\
	=
	 \nabla_\thetavecsmall^T 
	\left\{{\rm{E}}_\thetavecsmall \left[\hat{\theta}_m| \xvec \in {\mathcal{A}}_m \right] \right\}
	\sum_{\alphavecsmall\in{\mathcal{A}}_m}  \Pr(\xvec=\alphavec| \xvec \in {\mathcal{A}}_m  ;\thetavec) 
	\nonumber\\
	= \nabla_\thetavecsmall^T 
	\Big\{{\rm{E}}_\thetavecsmall \left[\hat{\theta}_m| \xvec \in {\mathcal{A}}_m \right] \Big\}  , \hspace{4.1cm}
	\eeqna
	where the last equality stems  from the fact that for a conditional pmf we have
	$\sum_{\alphavecsmall\in{\mathcal{A}}_m}  \Pr(\xvec=\alphavec| m\in G_{N,0} (\xvec);\thetavec) =1$.
	By substituting the definition of the missing-mass bias vector, $\bvec_{N,0}(\thetavec)$, from (\ref{bvec}) in \eqref{51_B}
	and using the fact that $\nabla_\thetavecsmall^T \theta_m=\evec_m^T$, one obtains
	\beqna
	\label{endthild}
		\sum_{\alphavecsmall\in{\mathcal{A}}_m } \epsilon_m(\thetavec) \nabla_\thetavecsmall^T  \Pr(\xvec=\alphavec|\xvec \in {\mathcal{A}}_m  (\xvec);\thetavec)\hspace{1cm}
	\nonumber\\
	= \nabla^T_\thetavecsmall \bigg\{  \frac{\left[\bvec_{N,0}(\thetavec)\right]_m}{\Pr(\xvec \in {\mathcal{A}}_m ;\thetavec)}
	\bigg\}  
	+ 	\evec_m^T, 
	\eeqna
	$ m=1,\ldots,M$.
	Substitution of \eqref{Smat} in \eqref{endthild} and then substituting the result in \eqref{prod} results in \eqref{52}.
	\end{proof}

\subsection{Covariance inequality}
\label{cov_subsec}
The following part of the proof  is along the path of the proof from \cite{Stoica_Ng_1998}
for the CCRB on the MSE in a conventional estimation problem.
Let $\Wmat \in \mathbb{R}^{M\times M}$ be an arbitrary matrix and $\tilde{\thetavec}\in\Omega_\thetavecsmall$ is a local parameter vector. Then,
	\beqna
	\label{inequal}
	\zerovec \preceq {\rm{E}}_{\tilde{\thetavecsmall}} \left[
	\left( \Gammamat(\xvec,\tilde{\thetavec}) -\Wmat^T\Umat \Umat^T \Deltamat(\xvec,\tilde{\thetavec})
	\right)\right.\hspace{1.4cm}
	\nonumber\\
	\left. \times\left( \Gammamat(\xvec,\tilde{\thetavec}) -\Wmat^T\Umat \Umat^T \Deltamat(\xvec,\tilde{\thetavec})
	\right)^T \right] \hspace{1.35cm}
	\nonumber\\
	=
	{\rm{E}}_{\tilde{\thetavecsmall}} \left[ \Gammamat(\xvec,\tilde{\thetavec})\Gammamat^T(\xvec,\tilde{\thetavec}) \right]
	- \Smat^T(\tilde{\thetavec})\Umat \Umat^T \Wmat\hspace{0.9cm}
	\nonumber\\
	- \Wmat^T
	\Umat \Umat^T \Smat(\tilde{\thetavec}) 
	 + \Wmat^T
	 \Umat \Umat^T \Jmat^{(0)}(\tilde{\thetavec})\Umat \Umat^T \Wmat, 
	\eeqna
	where we substitute (\ref{52}) from Lemma \ref{lemma1} and the mmFIM definition from \eqref{JJJdef}.
	By rearranging \eqref{inequal}, we obtain
	\beqna
	\label{maxit}
	\zerovec \preceq
	{\rm{E}}_{\tilde{\thetavecsmall}}\left[\Gammamat(\xvec,\tilde{\thetavec})\Gammamat^T(\xvec,\tilde{\thetavec})\right]
-
	\Smat^T(\tilde{\thetavec})\Umat \Umat^T \Wmat\hspace{1cm}
			\nonumber\\
-\Wmat^T\Umat \Umat^T \Smat(\tilde{\thetavec})
	+\Wmat^T\Umat \Umat^T \Jmat^{(0)}(\tilde{\thetavec})\Umat \Umat^T \Wmat.
	\eeqna
	By applying the trace operator on \eqref{maxit}, it can be verified that the matrix inequality in \eqref{maxit} provides   a family of bounds on the mmMSE from \eqref{law}, which depends on the specific choice of the matrix $\Wmat$.  Theorem \ref{theoremCRB}
is obtained by choosing 
the optimal member from this family, as described in the following.

 Since $\Umat^T\Jmat^{(0)}(\tilde{\thetavec})\Umat$ is
	a non-singular matrix under regularity Condition \ref{cond1}, then
it is shown in  \cite{Stoica_Ng_1998}
that 
 the greatest lower bound, i.e.    the supremum of the r.h.s. of \eqref{maxit}
over $\Wmat$, is obtained by a matrix $\Wmat$ which satisfies
	\be
	\label{maxhold}
	\Wmat^T\Umat=\Smat^T(\tilde{\thetavec})\Umat\left(\Umat^T\Jmat^{(0)}(\tilde{\thetavec})\Umat\right)^{-1}.
	\ee
	By substituting (\ref{maxhold}) into (\ref{maxit}), one obtains
	\beqna
	\label{59}
	{\rm{E}}_{\tilde{\thetavecsmall}}\left[\Gammamat(\xvec,\tilde{\thetavec})\Gammamat^T(\xvec,\tilde{\thetavec})\right]\hspace{4cm}
	\nonumber\\
	\succeq
	\Smat^T(\tilde{\thetavec})\Umat(\Umat^T\Jmat^{(0)}(\tilde{\thetavec})\Umat)^{-1}\Umat^T\Smat(\tilde{\thetavec}) .
	\eeqna
	By applying the trace operator on \eqref{59} and substituting \eqref{Gamma_def}, we obtain
	\beqna
	\label{68}
	\sum_{m=1}^M
{\rm{E}}_{\tilde{\thetavecsmall}}
	\left[ \epsilon_m(\tilde{\thetavec}) ^2{\mathbbm{1}}_{\{ m \in G_{N,0} (\xvec) \}} \right]\hspace{2.75cm}
	\nonumber\\
	\geq
	{\text{trace}}(\Smat^T(\tilde{\thetavec})\Umat(\Umat^T\Jmat^{(0)}(\tilde{\thetavec})\Umat)^{-1}\Umat^T\Smat(\tilde{\thetavec})).
	\eeqna

	The following equation relates  the mmMSE from \eqref{law} and the l.h.s. of
	\eqref{68}.
From \eqref{epsilom},	it  can be seen that
	\beqna
	\label{connection}
	{\rm{E}}_{\tilde{\thetavecsmall}}\left[(\hat{\theta}_m -\tilde{\theta}_m)^2 | \xvec \in {\mathcal{A}}_m  \right]
	\hspace{3.45cm}
	\nonumber\\=
		{\rm{E}}_{\tilde{\thetavecsmall}}\left[(\epsilon_m(\tilde{\thetavec}) + {\rm{E}}_{\tilde{\thetavecsmall}} \left[\hat{\theta}_m| \xvec \in {\mathcal{A}}_m \right] -\tilde{\theta}_m)^2 |\xvec \in {\mathcal{A}}_m  \right]
		\nonumber\\
	=
		{\rm{E}}_{\tilde{\thetavecsmall}}\left[\epsilon_m^2(\tilde{\thetavec})  | \xvec \in {\mathcal{A}}_m  \right]+( {\rm{E}}_{\tilde{\thetavecsmall}} \left[\hat{\theta}_m| \xvec \in {\mathcal{A}}_m \right] -\tilde{\theta}_m)^2
		\nonumber\\
		=
		{\rm{E}}_{\tilde{\thetavecsmall}}\left[\epsilon_m^2(\tilde{\thetavec})  | \xvec \in {\mathcal{A}}_m  \right]+\left( 
\frac{\left[\bvec_{N,0}(\tilde{\thetavec}) \right]_m }{\Pr ( \xvec \in {\mathcal{A}}_m ;\tilde{\thetavec})}\right)^2
	,\hspace{0.25cm}
		\eeqna
			where the second equality stems from  \eqref{epsilonzero} and the last equality is obtained by substituting 
\eqref{bvec}.
By multiplying  \eqref{connection} by 
$\Pr (\xvec \in {\mathcal{A}}_m;\tilde{\thetavec})$,
then summing the result over  $ m=1,\ldots, M$, and substituting
the result in  \eqref{68}, one obtains
the mmBCRB on the mmMSE of biased estimator evaluated at the local point, $\tilde{\thetavec}$, in  \eqref{TheBound}-\eqref{CRB}.
	%%%%%%%

\subsection{Derivation of the equality condition in \eqref{CRB_equality}}
\label{equality_cond_proof}
According to Cauchy-Schwartz inequality properties (or covariance inequality properties),  
equality in  (\ref{inequal}) 
for  $\Wmat$ that satisfies (\ref{maxhold}) holds if 
	\beqna
	\label{zero_app1}
	{\rm{E}}_{\tilde{\thetavecsmall}} \big[(\Gammamat(\xvec,\tilde{\thetavec})-\Smat^T(\tilde{\thetavec})\Umat(\Umat^T\Jmat^{(0)}(\tilde{\thetavec})\Umat)^{-1} \Umat^T\Deltamat(\xvec,\tilde{\thetavec}))
	\nonumber\\
	\times
	(\Gammamat(\xvec,\tilde{\thetavec})-\Smat^T(\tilde{\thetavec})\Umat(\Umat^T\Jmat^{(0)}(\tilde{\thetavec})\Umat)^{-1} \Umat^T\Deltamat(\xvec,\tilde{\thetavec}))^T \big]
	\nonumber\\=\zerovec.
	\eeqna
	The condition in \eqref{zero_app1} holds if
	\be
	\Gammamat(\xvec,\tilde{\thetavec}) =\Smat^T(\tilde{\thetavec})\Umat(\Umat^T\Jmat^{(0)}(\tilde{\thetavec})\Umat)^{-1} \Umat^T\Deltamat(\xvec,\tilde{\thetavec})
	\ee
	which, by using \eqref{Gamma_def}, implies that
	\be
	\label{CRB_equality_app}
	\epsilon_m(\tilde{\thetavec}) = \left[\Smat(\tilde{\thetavec})^T\Umat(\Umat^T\Jmat^{(0)}(\tilde{\thetavec})\Umat)^{-1} \Umat^T\Deltamat(\xvec,\tilde{\thetavec}) \right]_{m,m},
	\ee
for any 	$ m=1,\ldots,M$ such that $m\in G_{N,0}\xvec $. 
	By substituting \eqref{bvec} and \eqref{epsilom} in \eqref{CRB_equality_app} we get \eqref{CRB_equality}.
	%%%%%%%%%%%%%%%%%%

\section{Derivation of the FIM in \eqref{matrixFIM}}
\label{FIMappendix}
In this appendix, we develop the closed-form mmFIM,   defined in \eqref{JJJdef},  for the observation model from \eqref{zero}. 
By substituting \eqref{bayesmm} in \eqref{delta}, one obtains
\beqna
\label{delta_bayes2}
\Deltamat_{1:M,m}(\xvec,\thetavec)= \hspace{5cm}\nonumber\\
\nabla_\thetavecsmall  
\left(\sum_{l=1}^M C_{N,l}(\xvec) \log  \theta_l 
-N \log (1 - \theta_m)\right)
{\mathbbm{1}}_{\{ m \in G_{N,0} (\xvec) \}}
\nonumber\\
=\left(	\vvec^T(\xvec,\thetavec) -\frac{N}{1 - \theta_m}\evec_m^T\right){\mathbbm{1}}_{\{ m \in G_{N,0} (\xvec) \}},\hspace{1.9cm}
\eeqna
where $\vvec(\xvec,\thetavec)$ is defined in \eqref{vvec}.
By substituting 
\eqref{delta_bayes2} in 
 the mmFIM definition from \eqref{JJJdef}, we obtain that the $(k,l)$th element of the mmFIM for our model is given by
\beqna
\label{JJ}
\left[ \Jmat^{(0)} (\thetavec) \right]_{k,l}
=\sum_{m=1}^M {\rm{E}}_\thetavecsmall \left[
\left[\Deltamat(\xvec,\thetavec)\right]_{k,m} \left[\Deltamat(\xvec,\thetavec)\right]_{l,m} \right]\hspace{0.75cm}
\nonumber\\
=
\sum_{m=1}^M
{\rm{E}}_\thetavecsmall\left[
\left( \frac{C_{N,k}(\xvec)}{\theta_k}+\frac{N}{1-\theta_m}\delta_{k,m} \right) \right.
\hspace{1cm}
\nonumber\\
\times
\left. \left( \frac{C_{N,l}(\xvec)}{\theta_l}+\frac{N}{1-\theta_m}\delta_{l,m} \right)
{\mathbbm{1}}_{\{ m \in G_{N,0} (\xvec) \}}
\right]
\nonumber\\
= 
\sum_{m=1}^M \left\{
\frac{1}{\theta_k\theta_l}
{\rm{E}}_{\thetavecsmall}  \left[ C_{N,k}(\xvec) C_{N,l}(\xvec) | m\in G_{N,0}(\xvec) \right]  \right.
\nonumber\\
+
\frac{N \delta_{k,m}}{(1-\theta_m)\theta_l}
{\rm{E}}_{\thetavecsmall}  \left[ C_{N,l}(\xvec) | \xvec \in {\mathcal{A}}_m  \right]\hspace{1cm}
\nonumber\\
+ \frac{N \delta_{l,m}}{(1-\theta_m)\theta_k}
{\rm{E}}_{\thetavecsmall}  \left[ C_{N,k}(\xvec)| \xvec \in {\mathcal{A}}_m  \right]\hspace{0.85cm}
\nonumber\\
+ \left.
\frac{N^2 \delta_{k,m}\delta_{l,m}}{(1-\theta_m)^2} 
\right\} \Pr(\xvec \in {\mathcal{A}}_m ;\thetavec), \hspace{0.75cm}
\eeqna	
where the last equality is obtained by using  the law of total probability.
In the following, we compute the conditional expectation terms from \eqref{JJ}.
First, 
it can be seen that
\beqna
\label{explain}
\sum_{\alphavecsmall\in{\mathcal{A}}_m } \Pr (\xvec = \alphavec;\thetavec )=
\Pr (\xvec \in {\mathcal{A}}_m ;\thetavec)
=
\left(\sum_{n=1,n\neq m}^N \theta_n\right)^N,
\eeqna
$m=1,\ldots,M$,
where ${\mathcal{A}}_m$ is defined in \eqref{Omega_def} and where the probability of the $m$th element  to be unobserved on the r.h.s. of \eqref{explain} is an alternative way of writing the r.h.s of 
\eqref{probofkunseen}.
Then, by using (\ref{bayesmm}), it can be verified that
\beqna
\label{howmoment1}
{\rm{E}}_{\thetavecsmall}  \left [ C_{N,k}(\xvec)  | \xvec \in {\mathcal{A}}_m   \right]\hspace{3.5cm}
\nonumber\\
=
\sum_{\alphavecsmall\in{\mathcal{A}}_m } C_{N,k}(\alphavec) 
\frac{\prod_{n=1}^M \theta_n^{C_{N,n}(\alphavecsmall)}}{(1-\theta_m)^N} \hspace{0.9cm}
\nonumber\\
= \frac{\theta_k}{(1-\theta_m)^N}\sum_{\alphavecsmall\in{\mathcal{A}}_m }  \frac{\partial}{\partial\theta_k}
\prod_{n=1}^M \theta_n^{C_{N,n}(\alphavecsmall)} 
\hspace{0.2cm}
\nonumber\\
= \frac{\theta_k}{(1-\theta_m)^N} \frac{\partial}{\partial\theta_k}
\sum_{\alphavecsmall\in{\mathcal{A}}_m } \Pr (\xvec = \alphavec;\thetavec ),
\eeqna
$m,k=1,\ldots,M$.
It should be noted that the last equality in \eqref{howmoment1} is only valid if we use
the probability term, $\Pr (\xvec = \alphavec;\thetavec )$, as it is written in \eqref{explain}.
By substituting \eqref{explain} in \eqref{howmoment1}, one obtains
\beqna
\label{howmoment1_add}
{\rm{E}}_{\thetavecsmall}  \left [ C_{N,k}(\xvec)  |\xvec \in {\mathcal{A}}_m   \right]\hspace{3cm}
\nonumber\\
= \frac{\theta_k}{(1-\theta_m)^N} \frac{\partial}{\partial\theta_k} \Big(\sum_{n=1,n\neq m}^N \theta_n\Big)^N 
\nonumber\\
=
\begin{cases}
	\frac{N \theta_k}{1-\theta_m}  & m \neq k \\
		0  &  m = k 
\end{cases}. \hspace{2.1cm}
\eeqna
Similar to the derivation of (\ref{howmoment1}), by using (\ref{bayesmm}), we obtain
 \beqna
 \label{Cnml}
{\rm{E}}_{\thetavecsmall}  \left [ C_{N,k}(\xvec)C_{N,l}(\xvec)  |\xvec \in {\mathcal{A}}_m   \right]
\hspace{2.8cm}
\nonumber\\
=
\sum_{\alphavecsmall\in{\mathcal{A}}_m } C_{N,k}(\alphavec)C_{N,l} (\alphavec)
\frac{\prod_{n=1}^M \theta_n^{C_{N,n}(\alphavec)}}{(1-\theta_m)^N} \hspace{2cm}
\nonumber\\
= \sum_{\alphavecsmall\in{\mathcal{A}}_m } \frac{\theta_k\theta_l}{(1-\theta_m)^N} \frac{\partial^2}{\partial\theta_k\partial\theta_l}
\prod_{n=1}^M \theta_n^{C_{N,n}(\alphavec)} \hspace{2cm}
\nonumber\\
+ 
\frac{\delta_{k,l}}{(1-\theta_m)^N} \sum_{\alphavecsmall\in{\mathcal{A}}_m } C_{N,k}(\alphavec)
\prod_{n=1}^M \theta_n^{C_{N,n}(\alphavec)} \hspace{1.75cm}
\nonumber\\
= \frac{\theta_k\theta_l}{(1-\theta_m)^N} \frac{\partial^2}{\partial\theta_k\partial\theta_l}
\sum_{\alphavecsmall\in{\mathcal{A}}_m } \Pr (\xvec = \alphavec;\thetavec )
\hspace{1.75cm}
\nonumber\\
+\frac{N\delta_{k,l}}{1-\theta_m}{\rm{E}}_{\thetavecsmall}  \left [ C_{N,k}(\xvec)  | \xvec \in {\mathcal{A}}_m   \right], \hspace{1.75cm}
\eeqna
where we replace the order of the derivative and the sum, and we use
\eqref{zero}. Again, it should be noted that the last equality in \eqref{Cnml} is only valid if we use 
the probability term, $\Pr (\xvec = \alphavec;\thetavec )$, in \eqref{explain}.
By substituting \eqref{explain} and \eqref{howmoment1_add} in \eqref{Cnml}, one obtains
\beqna
{\rm{E}}_{\thetavecsmall}  \left [ C_{N,k}(\xvec)C_{N,l}(\xvec)  | \xvec \in {\mathcal{A}}_m   \right]=
\hspace{3cm}
\nonumber\\
 \frac{\theta_k\theta_l}{(1-\theta_m)^N} \frac{\partial^2}{\partial\theta_k\partial\theta_l} \Big(\sum_{n=1,n\neq m}^N \theta_n\Big)^N
+ \frac{N\theta_k \delta_{k,l}(1-\delta_{m,k})}{1-\theta_m}
\nonumber\\
= 
\begin{cases}
	\frac{N(N-1) \theta_k\theta_l}{(1-\theta_m)^2}   &  {\text{if }}m \neq k,~ l \neq m,~ k \neq l\\
		\frac{N(N-1) \theta_k^2}{(1-\theta_m)^2} +  \frac{N \theta_k}{1-\theta_m} 
		 & {\text{if }} k = l \neq m \\
		0 & {\text{otherwise}}
\end{cases}.
\label{moment2}
\eeqna
Thus, by substituting \eqref{probofkunseen}, \eqref{howmoment1_add}, and \eqref{moment2} in \eqref{JJ}, one obtains that the elements of the mmFIM are given by
\beqna
\label{elements_kk}
[\Jmat^{(0)}(\thetavec)]_{k,k}
\hspace{5.5cm}
\nonumber\\
\sum_{m=1,m\neq k}^M \left(\frac{N(N-1)}{(1-\theta_m)^2} +  \frac{N}{\theta_k(1-\theta_m)} \right)(1-\theta_m)^N 
\nonumber\\
+ \frac{N^2}{(1-\theta_k)^2} (1-\theta_k)^N  \hspace{4cm}
\nonumber\\
= \sum_{m=1}^M \left(\frac{N(N-1)}{(1-\theta_m)^2} +  \frac{N}{\theta_k(1-\theta_m)} \right)(1-\theta_m)^N 
\nonumber\\
+ \frac{N}{(1-\theta_k)^2} (1-\theta_k)^N  \hspace{3.25cm}
\nonumber\\
- \frac{N}{\theta_k(1-\theta_k)} (1-\theta_k)^N , \hspace{3cm}
\eeqna
for $k=1,\ldots,M$,
and 
\beqna
\label{elements_k_neq_l}
[\Jmat^{(0)}(\thetavec)]_{k,l}
=
\sum_{m=1,m\neq k,m\neq l}^M  \frac{N(N-1)}{(1-\theta_m)^2}
(1-\theta_m)^N  \hspace{0.5cm}
\nonumber\\
+ \frac{N^2}{(1-\theta_k)^2} (1-\theta_k)^N   + \frac{N^2}{(1-\theta_l)^2} (1-\theta_l)^N  \hspace{1.5cm}
\nonumber\\
 = \sum_{m=1}^M \frac{N(N-1)}{(1-\theta_m)^2} (1-\theta_m)^N 
 \nonumber\\
 + \frac{N}{(1-\theta_k)^2} (1-\theta_k)^N  + \frac{N}{(1-\theta_l)^2} (1-\theta_l)^N ,
\eeqna
for $k,l=1,\ldots,M$, $k\neq l$.
By rearranging the elements in \eqref{elements_kk} and \eqref{elements_k_neq_l}, and substituting \eqref{probofkunseen}, we obtain the closed-form matrix $\Jmat^{(0)}(\thetavec)$ in its matrix representation in \eqref{matrixFIM}.
\beqna
D_{m,m}= - \frac{(1-\theta_m)^N  }{(1-\theta_m)^2} 
 +\sum_{l=1,l\neq m}^M   \frac{(1-\theta_l)^N }{\theta_m(1-\theta_l)}. 
\eeqna
%%%%%%%%%%%%%%%%%%%%%%%%%
\section{Proof of Lemma \ref{Smat_theorem}}
\label{proofSmat}
Taking the logarithm of \eqref{bayesmm} yields 
the following conditional log-likelihood function:
\beqna
\label{cond_log_zero}
\log p(\xvec | \xvec \in {\mathcal{A}}_m  ;\thetavec)=
\sum_{l=1}^M C_{N,l}(\xvec) \log \theta_l-N\log (1 - \theta_m),
\eeqna
for any $\xvec \in {\mathcal{A}}_m$.
By substituting \eqref{epsilom} and the derivative of  \eqref{cond_log_zero}  w.r.t. $\thetavec$ in \eqref{prod}, we obtain
\beqna
	\label{prod_after_iid}
	{\rm{E}}_{\thetavecsmall} \left[ \epsilon_m(\thetavec)  \Deltamat_{1:M,m}^T(\xvec,\thetavec)
	\right]
	=	{\rm{E}}_{\thetavecsmall} \left[ (\hat{\theta}_m -{\rm{E}}_\thetavecsmall [\hat{\theta}_m| \xvec \in {\mathcal{A}}_m ])
	\right.
	\nonumber\\
	\left.
\times	(\vvec^T(\xvec,\thetavec)+N\evec_m^T\frac{1}{1 - \theta_m}){\mathbbm{1}}_{\{ \xvec \in {\mathcal{A}}_m  \}}\right]
		\nonumber\\
		=		{\rm{E}}_{\thetavecsmall} \left[ (\hat{\theta}_m -\theta_m)	(\vvec^T(\xvec,\thetavec)+N\evec_m^T\frac{1}{1 - \theta_m}){\mathbbm{1}}_{\{ \xvec \in {\mathcal{A}}_m  \}}\right]
		\nonumber\\
-[\bvec_{N,0}(\thetavec) ]_m \left({\rm{E}}_{\thetavecsmall} \left[
	\vvec^T(\xvec,\thetavec)| \xvec \in {\mathcal{A}}_m \right]+N\evec_m^T\frac{1}{1 - \theta_m}\right),
\eeqna
where the last equality is obtained by substituting 
\eqref{bvec}.
By substituting \eqref{bvec} and \eqref{howmoment1_add} in \eqref{prod_after_iid}, we obtain
\beqna
	\label{prod_after_iid2}
	{\rm{E}}_{\thetavecsmall} \left[ \epsilon_m(\thetavec)  \Deltamat_{1:M,m}^T(\xvec,\thetavec)
	\right]\hspace{4.25cm}
	\nonumber\\
		=		{\rm{E}}_{\thetavecsmall} \left[ (\hat{\theta}_m -\theta_m)	\vvec^T(\xvec,\thetavec){\mathbbm{1}}_{\{ \xvec \in {\mathcal{A}}_m  \}}\right]\hspace{2.25cm}
				\nonumber\\
		+\frac{N}{1 - \theta_m}[\bvec_{N,0}(\thetavec) ]_m\evec_m^T
-N\frac{1}{1 - \theta_m}[\bvec_{N,0}(\thetavec) ]_m \onevec_M^T,
\eeqna
$m=1,\ldots,M$. 
By substituting \eqref{prod_after_iid2} in \eqref{52},
we obtain \eqref{Smat_final}.
	%%%%%%%
\small
% Generated by IEEEtran.bst, version: 1.14 (2015/08/26)

\end{document}

%% file: myshorts.tex
\newcommand{\avec}{{\bf{a}}}

\newcommand{\bvec}{{\bf{b}}}
\newcommand{\cvec}{{\bf{c}}}

\newcommand{\evec}{{\bf{e}}}

\newcommand{\yvec}{{\bf{y}}}

\newcommand{\wvec}{{\bf{w}}}
\newcommand{\xvec}{{\bf{x}}}

\newcommand{\vvec}{{\bf{v}}}

\newcommand{\etavec}{{\bf{\eta}}}

\newcommand{\onevec}{{\bf{1}}}
\newcommand{\zerovec}{{\bf{0}}}

\newcommand{\alphavec}{{\bf{\alpha}}}

\newcommand{\Gammamat}{{\bf{\Gamma}}}
\newcommand{\Amat}{{\bf{A}}}

\newcommand{\Dmat}{{\bf{D}}}

\newcommand{\Fmat}{{\bf{F}}}

\newcommand{\Jmat}{{\bf{J}}}
\newcommand{\Imat}{{\bf{I}}}

\newcommand{\Pmat}{{\bf{P}}}

\newcommand{\Smat}{{\bf{S}}}

\newcommand{\Umat}{{\bf{U}}}

\newcommand{\Wmat}{{\bf{W}}}

\newcommand{\Deltamat}{{\bf{\Delta}}}

\newcommand{\define}{\stackrel{\triangle}{=}}

\def\alphavec{{\mbox{\boldmath $\alpha$}}}

\def\etavec{{\mbox{\boldmath $\eta$}}}

\def\thetavec{{\mbox{\boldmath $\theta$}}}

\def\thetavecsmall{{\mbox{\boldmath {\scriptsize $\theta$}}}}

\def\alphavecsmall{{\mbox{\boldmath {\scriptsize $\alpha$}}}}

\newcommand{\be}{\begin{equation}}
\newcommand{\ee}{\end{equation}}
\newcommand{\beqna}{\begin{eqnarray}}
\newcommand{\eeqna}{\end{eqnarray}}